\documentclass[a4paper,onecolumn,superscriptaddress,accepted=2024-03-04,10pt]{quantumarticle}

\pdfoutput=1

\usepackage[T1]{fontenc} 
\usepackage[english]{babel} 
\usepackage[utf8]{inputenc}

\usepackage[margin=1.8cm]{geometry}

\usepackage[numbers,sort&compress]{natbib}

\providecommand{\prl}{Phys.\ Rev.\ Lett.}

\usepackage[table,usenames,dvipsnames]{xcolor}
\usepackage[
		colorlinks,
		linkcolor={black!30!blue},
		citecolor={black!30!blue},
		urlcolor={black!30!blue}
]{hyperref}
\usepackage{graphicx}
\usepackage{hyperxmp} 

\usepackage{amsmath,amssymb,mathtools,amsthm,dsfont}
\usepackage{stmaryrd}

\usepackage{enumitem}
\usepackage{subfigure}
\usepackage[normalem]{ulem}

\usepackage{optidef}
\usepackage{algorithm,algorithmicx}
\usepackage[noend]{algpseudocode}

\usepackage[capitalize]{cleveref}
\crefname{equation}{Eq.}{Eqs.}
\Crefname{equation}{Equation}{Equations}

\usepackage[printonlyused,
						withpage,
						]{acronym}
\makeatletter
\AtBeginDocument{%
 \renewcommand*{\AC@hyperlink}[2]{%
   \begingroup
     \hypersetup{hidelinks}%
     \hyperlink{#1}{#2}%
   \endgroup
 }%
}
\makeatother

\newtheorem{theorem}{Theorem}
\newtheorem*{theorem*}{Theorem}

\newtheorem{lemma}[theorem]{Lemma}
\newtheorem{proposition}[theorem]{Proposition}

\newtheorem{problem}[theorem]{Problem}
\newtheorem{definition}[theorem]{Definition}
\newtheorem{conjecture}[theorem]{Conjecture}
\newtheorem*{conjecture*}{Conjecture}
\theoremstyle{definition}

\newcommand{\e}{\ensuremath\mathrm{e}}
\renewcommand{\i}{\ensuremath\mathrm{i}}

\DeclareMathOperator{\LandauO}{\mathrm{O}}

\newcommand{\fro}{\mathrm{F}}

\newcommand{\Gate}[1]{\operatorname{#1}}

\NewDocumentCommand\Cl{mg}{
    \ensuremath{\mathrm{Cl}_{#1}\IfNoValueTF{#2}{}{(#2)}}%
}
\NewDocumentCommand\HW{mg}{
    \ensuremath{\mathrm{HW}_{#1}\IfNoValueTF{#2}{}{(#2)}}%
}

\newcommand{\RR}{\mathbb{R}}
\newcommand{\QQ}{\mathbb{Q}}
\newcommand{\ZZ}{\mathbb{Z}}
\newcommand{\NN}{\mathbb{N}}
\newcommand{\1}{\mathds{1}}

\newcommand{\FF}{\mathbb{F}}

\newcommand{\R}{\mathbb{R}}

\renewcommand{\vec}[1]{\boldsymbol{#1}}

\DeclarePairedDelimiterX{\abs}[1]{\lvert}{\rvert}{%
  \ifblank{#1}{\,\cdot\,}{#1}
}   

\DeclarePairedDelimiterX\norm[1]\lVert\rVert{%
  \ifblank{#1}{\,\cdot\,}{#1}
}   

\newcommand{\lpnorm}[2][p]{\norm{#2}_{\ell_{#1}}}   

\newcommand{\lonenorm}[2][1]{\norm{#2}_{\ell_{#1}}}
\newcommand{\linfnorm}[2][\infty]{\norm{#2}_{\ell_{#1}}}
\DeclarePairedDelimiterX{\iiiNorm}[1]{\lvert}{\rvert}{%
  \delimsize\lvert\delimsize\lvert#1\delimsize\rvert\delimsize\rvert%
}

\DeclarePairedDelimiterXPP\snorm[1]{}\lVert\rVert{_\infty}{\ifblank{#1}{\,\cdot\,}{#1}}   

\DeclarePairedDelimiterXPP\twonorm[1]{}\lVert\rVert{_2}{\ifblank{#1}{\,\cdot\,}{#1}}   

\DeclarePairedDelimiterXPP\trnorm[1]{}\lVert\rVert{_1}{\ifblank{#1}{\,\cdot\,}{#1}}   

\DeclarePairedDelimiterXPP\fnorm[1]{}\lVert\rVert{_{\fro}}{\ifblank{#1}{\,\cdot\,}{#1}}   

\DeclarePairedDelimiterXPP\dnorm[1]{}\lVert\rVert{_\diamond}{\ifblank{#1}{\,\cdot\,}{#1}}   

\DeclarePairedDelimiterXPP\cbnorm[1]{}\lVert\rVert{_\mathrm{cb}}{\ifblank{#1}{\,\cdot\,}{#1}}   
\DeclarePairedDelimiterXPP\onenorm[1]{}\lVert\rVert{_{1\rightarrow 1}}{\ifblank{#1}{\,\cdot\,}{#1}}   
\DeclarePairedDelimiterXPP\ddnorm[1]{}\lVert\rVert{_{\diamond\rightarrow \diamond}}{\ifblank{#1}{\,\cdot\,}{#1}}   
\DeclarePairedDelimiterXPP\ssnorm[1]{}\lVert\rVert{_{\infty\rightarrow\infty}}{\ifblank{#1}{\,\cdot\,}{#1}}   

\providecommand\given{}
\newcommand\SetSymbol[1][]{%
  \nonscript\:#1\vert
  \allowbreak
  \nonscript\:
  \mathopen{}}
\DeclarePairedDelimiterX\Set[1]\{\}{%
  \renewcommand\given{\SetSymbol[\delimsize]}
  #1
}

\DeclarePairedDelimiterX\innerp[2]{\langle}{\rangle}{%
  \ifblank{#1}{\,\cdot\,}{#1} , \ifblank{#2}{\,\cdot\,}{#2}%
}

\DeclarePairedDelimiterX\braket[2]{\langle}{\rangle}%
  {#1\kern0.15ex\delimsize\vert\kern0.15ex\mathopen{}#2}

\DeclarePairedDelimiterX\ketbra[2]{\vert}{\vert}%
  {#1\kern0.15ex\delimsize\rangle\delimsize\langle\kern0.15ex\mathopen{}#2}

\DeclarePairedDelimiterX\sandwich[3]{\langle}{\rangle}%
  {#1\,\delimsize\vert\kern0.15ex\mathopen{}#2\kern0.15ex\delimsize\vert\kern0.15ex\mathopen{}#3}

\DeclarePairedDelimiterX\obraket[2]{(}{)}%
  {#1\kern0.15ex\delimsize\vert\kern0.15ex\mathopen{}#2}

\DeclarePairedDelimiterX\oketbra[2]{\vert}{\vert}%
  {#1\kern0.15ex\delimsize)\delimsize(\kern0.15ex\mathopen{}#2}

\DeclarePairedDelimiterX\osandwich[3]{(}{)}%
  {#1\,\delimsize\vert\kern0.15ex\mathopen{}#2\kern0.15ex\delimsize\vert\kern0.15ex\mathopen{}#3}

\newcommand{\myleft}{\mathopen{}\mathclose\bgroup\left}
\newcommand{\myright}{\aftergroup\egroup\right}

\DeclareMathOperator{\HTL}{\mathrm{Sym}_0}

\newcommand{\class}[1]{{\ensuremath{\mathsf{#1}}}}
\newcommand{\NP}{\class{NP}}

\renewcommand{\Gate}[1]{\operatorname{\mathsf{#1}}}

\newcommand{\hhu}{
	Institute for Theoretical Physics,
      Heinrich Heine University D{\"u}sseldorf, 
      Germany
}
\newcommand{\tuhh}{
    Institute for Quantum Inspired and Quantum Optimization,
    Hamburg University of Technology, Germany
}

\begin{document}

\title{Time-optimal multi-qubit gates: Complexity, efficient heuristic and  gate-time bounds}

\author{P.\ Baßler}
\email{bassler@hhu.de}
\affiliation{\hhu}
\author{M.\ Heinrich}
\affiliation{\hhu}
\author{M.\ Kliesch}
\email{martin.kliesch@tuhh.de}
\affiliation{\tuhh}

\begin{abstract}
Multi-qubit entangling interactions arise naturally in several quantum computing platforms and promise advantages over traditional two-qubit gates.
In particular, a fixed multi-qubit Ising-type interaction together with single-qubit $\Gate{X}$-gates can be used to synthesize global $\Gate{ZZ}$-gates ($\Gate{GZZ}$ gates).
In this work, we first show that the synthesis of such quantum gates that are time-optimal is \NP-hard.
Second, we provide explicit constructions of special time-optimal multi-qubit gates. 
They have constant gate times and can be implemented with linearly many $\Gate{X}$-gate layers.
Third, we develop a heuristic algorithm with polynomial runtime for synthesizing fast multi-qubit gates.
Fourth, we derive lower and upper bounds on the optimal $\Gate{GZZ}$ gate-time.
Based on explicit constructions of $\Gate{GZZ}$ gates and numerical studies, we conjecture that any $\Gate{GZZ}$ gate can be executed in a time $\LandauO(n)$ for $n$ qubits.
Our heuristic synthesis algorithm leads to $\Gate{GZZ}$ gate-times with a similar scaling, which is optimal in this sense.
We expect that our efficient synthesis of fast multi-qubit gates allows for faster and, hence, also more error-robust execution of quantum algorithms.
\end{abstract}

\maketitle

 \hypersetup{
	     pdftitle = {Time-optimal multi-qubit gates: Complexity, efficient heuristic and  gate-time bounds},
	     pdfauthor = {Pascal Baßler,
		     	Markus Heinrich,
	     		Martin Kliesch}
	     pdfsubject = {Quantum computing},
	     pdfkeywords = {quantum, compiling,
	     	 global, interaction, entangling, multi-qubit, gate, synthesis, unitary, all-to-all, connectivity, gate-time, time-optimal,
		     ion trap, Ising, Hamiltonian, NISQ,
		     convex, optimization, linear program, LP, MIP, mixed integer program,
		     complexity, NP, heuristic, cut, polytope, graph theory,
		     digital-analog, DAQC,
		     qubit,
		     Molmer-Sorensen
	     }
	    }

\section{Introduction}
Any quantum computation requires to decompose its logical operations into the platform's native instruction set.
The performance of the computation depends heavily on the available instructions and their implementation in the quantum hardware.
In particular for early quantum devices a major challenge is posed by their short decoherence times, which limits the runtime of a quantum computation significantly.
Therefore, it is not only necessary to optimize the number of native instructions but also their execution time. 

Ising-type interactions give rise to an important and rich class of Hamiltonians ubiquitous in several quantum computing platforms \cite{wang_multibit_2001,monz_14-qubit_2011,kjaergaard_superconducting_2020,figgatt_parallel_2019,lu_scalable_2019}.
Previously, we have utilized these Ising-type interactions in a new synthesis method \cite{basler_synthesis_2022}.
In particular, we have considered the problem of synthesizing time-optimal multi-qubit gates on a quantum computing platform that supports the following basic operations:
\begin{enumerate}[label= (\Roman*)]\itemsep=0pt
	\item single-qubit rotations can be executed in parallel, and \label{item:parallel_1qubit}
	\item it offers a fixed Ising-type interaction. \label{item:Ising}
\end{enumerate}
The corresponding synthesis of global $\Gate{ZZ}$-gates ($\Gate{GZZ}$ gates) is given by the minimization of the overall gate time, which can be written as a \ac{LP} \cite{basler_synthesis_2022}.
This \ac{LP} is exponentially large in the number of qubits. 
It has been unclear whether such multi-qubit gates can be synthesized in a computationally efficient way while keeping the gate time optimal.

In this work, we prove that this synthesis problem is \NP-hard and provide a close-to-optimal efficient heuristic solution.
To establish this hardness result, we draw the connection between the synthesis problem and graph theory.
The \emph{cut polytope} is defined as the convex hull of binary vectors representing the possible cuts of a given graph \cite{Barahona1986}.
We provide a polynomial time reduction of the membership problem of the cut polytope to the synthesis of time-optimal multi-qubit gates.
Since this membership problem is \NP-complete \cite{Garey_2002_Computers}, the synthesis of time-optimal multi-qubit gates is \NP-hard.
This is akin to the \NP-hardness of finding an optimal control pulse for multi-qubit gates using the Mølmer-Sørensen mechanism \cite{figgatt_parallel_2019,lu_scalable_2019}.

We provide several ways to circumvent the hardness of time-optimal $\Gate{GZZ}$-gates.
First, we provide explicit constructions of time-optimal multi-qubit gates realizing nearest-neighbor coupling under physically motivated assumptions.
Such constructed nearest-neighbor multi-qubit gates exhibit a constant gate time and can be implemented with only linearly many single-qubit gate layers.
We then use these ideas to define relaxations of the underlying linear program, leading to a hierarchy of polynomial-time algorithms for the synthesis of fast multi-qubit gates.
By increasing the level in the hierarchy, this heuristic approach can be adapted to provide substantially better approximations to the optimal solution at the cost of higher polynomial runtime.
For a small number of qubits, numerical experiments show that the so-obtained gate times are close to the optimal solution and come with significant runtime savings.

Among others, we prove bounds on the minimal multi-qubit gate time, and conjecture that it scales at most linear with the number of qubits.
This claim is supported by a class of explicit constructions of time-optimal multi-qubit gates achieving the linear upper bound.
Moreover, we provide numerical evidence that these explicit solutions in fact yield the longest gate time for a small number of qubits.

We expect our results to be useful for the implementation of time-optimal multi-qubit gates in \ac{NISQ} devices and beyond.
The polynomial-time heuristic algorithm makes it possible to efficiently synthesize fast multi-qubit gates for a growing number of qubits.
The here considered multi-qubit gates have been useful in a number of different applications, some of which we investigated in previous work \cite{basler_synthesis_2022}.
Furthermore, they are the main building blocks for a class of \ac{IQP} circuits which might be classically hard to simulate \cite{bremner_16}.
More recently, it was shown that quantum memory circuits and boolean functions can be implemented with a constant number of $\Gate{GZZ}$ gates and additional ancilla qubits \cite{allcock_constant_depth_2023}.
Moreover, there is an ancilla-free construction of multi-qubit Clifford circuits using at most $26$ $\Gate{GZZ}$ gates \cite{bravyi_constant_cost_2022}.
As noted in Ref.~\cite{bravyi_constant_cost_2022}, there is also a shorter implementation for $n\leq 2^{12}$, requiring only $2(\log_2(n)+1)$ $\Gate{GZZ}$ gates.
This implementation is based on a decomposition in Ref.~\cite{bravyi_hadamard_free_2020} and the $\log$-depth implementation of a $\Gate{CX}$ circuit in formula (4) of Ref.~\cite{maslov_depth_2022}.

This paper is structured as follows: We first give a brief introduction to the synthesis of time-optimal multi-qubit gates \cite{basler_synthesis_2022}.
In \cref{sec:LP_is_NP} we proof that the time-optimal multi-qubit synthesis problem is \NP-hard.
However, in \cref{sec:ana_GZZ} we explicitly construct a certain class of time-optimal multi-qubit gates with constant gate time.
The heuristic algorithm based on the ideas of the previous section is introduced in \cref{sec:LP_approx}.
\Cref{sec:lin_scaling} provides gate time bounds for time-optimal multi-qubit gates.
Finally, \cref{sec:numerics} presents numerical evidence for the conjectured linear gate-time scaling and numerical benchmarks for the heuristic algorithm.

\section{Synthesizing multi-qubit gates with Ising-type interactions}
\label{sec:prev_work}
In this section, we give a short introduction to the synthesis of time-optimal multi-qubit gates. For more details we refer to the first two sections of Ref.~\cite{basler_synthesis_2022}.

On the abstract quantum computing platform with $n$ qubits specified by the requirements \ref{item:parallel_1qubit} and \ref{item:Ising} above, interactions between the qubits are generated by an Ising-type Hamiltonian
\begin{equation}
H_{ZZ} \coloneqq - \sum_{i<j}^n J_{ij} \sigma_z^i \sigma_z^j \, ,
\end{equation}
where $\sigma_z^i$ is the Pauli-Z operator acting on the $i$-th qubit.
Note, that diagonal terms, where $i=j$, are excluded since they only change the Hamiltonian by an energy offset.
By $J$ we denote the symmetric matrix with entries $J_{ij}$ in the upper triangular part and vanishing diagonal.
We call $J$ the \emph{physical coupling matrix}.

Conjugating the Hamiltonian $H_{ZZ}$ with $\Gate{X}$ gates on the qubits indicated by the binary vector $\vec{b} \in \FF_2^n$ yields
\begin{equation}
\label{eq:Hising}
\begin{aligned}
\sigma_x^{\vec b} H_{ZZ} \sigma_x^{\vec b} &= - \sum_{i<j}^n J_{ij} \sigma_x^{b_i} \sigma_x^{b_j} \sigma_z^i \sigma_z^j \sigma_x^{b_i}\sigma_x^{b_j} \\
&= - \sum_{i<j}^n J_{ij} (-1)^{b_i} (-1)^{b_j} \sigma_z^i \sigma_z^j \\
&= - \sum_{i<j}^n J_{ij} m_i m_j \sigma_z^i \sigma_z^j \eqqcolon H(\vec m)  \, ,
\end{aligned}
\end{equation}
where we define the \emph{encoding} $\vec m \coloneqq (-1)^{\vec b}$ entry wise.
The sign of the interaction between qubit $i$ and $j$ is given by $m_i m_j \in \{-1,+1 \}$.
We call $H(\vec m)$ the \emph{encoded Hamiltonian}.

Given time steps $\lambda_{\vec m} \geq 0$ during which the encoding $\vec m$ is used, we consider unitaries of the form
\begin{equation}\label{eq:multi-qubitGate}
 \prod_{\vec m} \e^{-\i \lambda_{\vec m} H(\vec m)}
 = \e^{-\i \sum_{\vec m} \lambda_{\vec m} H(\vec m)}
 \eqqcolon \e^{-\i H} \, ,
\end{equation}
where we used that the diagonal Hamiltonians $H(\vec m)$ mutually commute.
For all possible encodings $\vec m\in \{-1,+1\}^n$ we collect the time steps $\lambda_{\vec m}$ in a vector $\vec \lambda \in \R^{2^n}$ and interpret $t= \sum_{\vec m}\lambda_{\vec m}$ as the total gate time of the unitary $\e^{-\i H}$, implemented by the sequence of unitaries \eqref{eq:multi-qubitGate}.
Moreover, we use the symmetry
\begin{equation}
\label{eq:m_sym}
 (-\vec m)(-\vec m)^T = \vec m \vec m^T
\end{equation}
to reduce the degrees of freedom in $\vec\lambda$ from $2^n$ to $2^{n-1}$ by adding up $\lambda_{\vec m} + \lambda_{-\vec m}$ to a single time step.

The so generated unitary is the time evolution operator under the \emph{total Hamiltonian}
\begin{equation}
\label{eq:targetHamiltonian}
  H = - \sum_{i<j}^n A_{ij} \sigma_z^i \sigma_z^j \, ,
\end{equation}
where we have defined the \emph{total coupling matrix}
\begin{equation}
\label{eq:Jdecomp}
A \coloneqq J \circ \sum_{\vec m} \lambda_{\vec m} \vec m \vec m^T \, ,
\end{equation}
and used the linearity of the Hadamard (entry-wise) product $\circ$.
By construction, $A$ is a symmetric matrix with vanishing diagonal.
Let us define the $\binom{n}{2}$-dimensional subspace of symmetric matrices with vanishing diagonal by
\begin{equation}\label{def:symtraceless}
	\HTL(\RR^n) \coloneqq \Set*{M \in \mathrm{Sym}(\RR^n)\given M_{ii}= 0 \ \forall i\in [n]} \, .
\end{equation}
For $A \in \HTL(\RR^n)$, we define an associated multi-qubit gate $\Gate{GZZ}(A)$, where $\Gate{GZZ}$ stands for ``global $ZZ$ interactions'',
\begin{equation}
\label{eq:GZZ}
\begin{aligned}
 \Gate{GZZ}(A) &\coloneqq \exp \left(\i \sum_{i<j}^n A_{ij} \sigma_z^i \sigma_z^j \right) \, .
\end{aligned}
\end{equation}

To determine which matrices can be decomposed as in \cref{eq:Jdecomp}, we denote the non-zero index set of a symmetric matrix as $\operatorname{nz}(M) \coloneqq \{ (i,j) \, | \, M_{ij}\neq 0 , i<j \}$.
Then, the subspace of matrices $A\in\HTL(\RR^n)$ that can be decomposed as in \cref{eq:Jdecomp} is exactly given by the condition $\operatorname{nz}(A)\subseteq\operatorname{nz}(J)$, which we assume to hold from now on.
Thus, all-to-all connectivity enables to decompose any coupling matrix $A\in\HTL(\RR^n)$ but is not strictly required by our approach. 
We call the number of encodings $\vec m$ needed for the decomposition the \emph{encoding cost} of $\Gate{GZZ}(A)$, and $\sum_{\vec m}\lambda_{\vec m}$ the \emph{total $\Gate{GZZ}$ time}.
Note, that both quantities depend on the chosen decomposition.

It is convenient to abstract the following analysis from the physical details given by $J$.
For a matrix $A\in\HTL(\RR^n)$ with $\operatorname{nz}(A)\subseteq\operatorname{nz}(J)$, its possible decompositions are in one-to-one correspondence with the decompositions of the matrix $M \coloneqq A \oslash J \in \HTL(\RR^n)$ where
\begin{equation}\label{eq:MasQuotient}
	(A \oslash J)_{ij} \coloneqq \begin{cases}
		A_{ij} / J_{ij}, & i \neq j \text{ and } J_{ij} \neq 0\, , \\
		0, & \text{otherwise}\, .
	\end{cases}
\end{equation}
We further define the linear operator 
\begin{equation}
\label{eq:lin_op}
\begin{aligned}
		\mathcal V: \RR_{\geq 0}^{2^{n-1}} &\rightarrow \HTL(\RR^n),
	\\
	\vec \lambda &\mapsto \sum_{\vec m} \lambda_{\vec m} \vec m \vec m^T ,
\end{aligned}
\end{equation}
represented in the standard basis by a matrix 
\begin{equation}
 	V  \in \{-1,+1 \}^{\binom{n}{2} \times (2^{n-1}) }\, .
 \end{equation}
Let $\vec v: \HTL (\RR^n) \rightarrow \RR^{\binom{n}{2}}$ be the (row-wise) vectorization of the upper triangular part of the matrix input such that the columns of $V$ are given by $\vec v(\vec m \vec m^T)$.
Our objective is to minimize the total gate time and the amount of $\vec m$'s needed to express the matrix $M \in \HTL(\RR^n)$.
To this end we formulate the following \acf{LP}:
\begin{mini}
{}{\vec 1^T \vec \lambda}{}{}
\label{eq:LP1}
\addConstraint{V \vec \lambda} {=\vec v(M)}{}{}
\addConstraint{ \vec \lambda}{ \in \R^{2^{n-1}}_{\geq 0}}{}{} \, ,
\end{mini}
where $\vec 1=(1,1,\dots,1)$ is the all-ones vector such that $\vec{1}^T \vec \lambda = \sum_{\vec m}\lambda_{\vec m}$.
A \emph{feasible solution} is an assignment of the variables that fulfills all constraints of the optimization problem and an \emph{optimal solution} is a feasible solution which also minimizes/maximizes the objective function.
Throughout this paper, we indicate optimal solutions by an asterisk, e.g.\ $\vec\lambda^*$.
Note, that the \ac{LP}~\eqref{eq:LP1} has a feasible solution for any symmetric matrix $M$ with vanishing diagonal \cite[Theorem~II.2]{basler_synthesis_2022}.
The theory of linear programming then guarantees the existence of an optimal solution with at most $\binom{n}{2}$ non-zero entries \cite[Proposition~II.3]{basler_synthesis_2022}.

A standard tool in convex optimization is \emph{duality} \cite{Boyd2009convexoptimization} which will be used in \cref{sec:lin_scaling}.
The dual \ac{LP} to the \ac{LP} \eqref{eq:LP1} reads as follows:
\begin{maxi}
{}{\innerp{M}{\vec y}}{}{}
\label{eq:dualLP}
\addConstraint{V^T \vec y } {\leq \vec 1}{}{}
\addConstraint{ \vec y }{\in \R^{\binom{n}{2}}}{}{} \, ,
\end{maxi}
with the inner product $\innerp{}{}: \HTL(\RR^n) \times \RR^{\binom{n}{2}} \rightarrow \RR$, $\innerp{M}{\vec y} \mapsto \vec v (M)^T \vec y $.
Here, inequalities between vectors are to be understood entry-wise. 
A simple, but important fact is the following:
If $\vec\lambda^*$ is an optimal solution to the primal \ac{LP} \eqref{eq:LP1}, then any feasible solution $\vec y$ to the dual \ac{LP} \eqref{eq:dualLP} provides a \emph{lower bound} as $\innerp{M}{\vec y} \leq \vec 1^T\vec\lambda^*$.
Moreover, the feasibility of the \ac{LP} \eqref{eq:LP1} implies that we have \emph{strong duality}:
if $\vec y^*$ is a dual optimal solution, then we have equality between the optimal values, $\innerp{M}{\vec y^*} = \vec 1^T\vec\lambda^*$.

\section{Main results}
\label{sec:main_results}
We want to highlight our main contributions.
First, we provide the hardness results for the synthesis of time-optimal $\Gate{GZZ}$ gates.
\begin{theorem*}[\cref{thm:hardness}]
The decision version of the \ac{LP}~\eqref{eq:LP1}, is \NP-complete.
\end{theorem*}
We say that the synthesis of time-optimal multi-qubit gates (solving \ac{LP}~\eqref{eq:LP1}) is \NP-hard in the sense of the function problem extension of the decision problem class \NP\,\cite{rich_2007}.
We circumvent the hardness of the time-optimal multi-qubit gate synthesis by providing an explicit construction of $\Gate{GZZ}$ gates realizing next neighbor coupling with constant gate time and linear encoding cost.
The assumption of a constant subdiagonal of $J$ is physically motivated and can be realized in an ion trap \cite{johanning_isospaced_2016}.
\begin{theorem*}[\cref{cor:apdx_subdiagonal}, informal]
 Let the subdiagonal of $J$ and $A$ be a constant with values $c$ and $\varphi$ respectively.
 Then $\Gate{GZZ}(A)$ on $n$ qubits has the encoding cost of $d \leq 2n$ and constant total $\Gate{GZZ}$ time $2 \varphi / c$.
 This total gate time is optimal.
\end{theorem*}
In \cref{sec:LP_approx} we define \cref{alg:eff_heu} and introduce \ac{LP}~\eqref{eq:LP_poly_hierarchy}, a polynomial time heuristic to synthesize $\Gate{GZZ}$ gates with small, but not necessarily minimal, gate time.
This heuristic does not rely on any further assumptions and is applicable for arbitrary $J,A \in \HTL(\RR^n)$.
In \cref{sec:numerics} we show numerically that this heuristic leads to $\Gate{GZZ}$ gate times close to the optimum while solving the synthesis problem much faster.

Furthermore, we proof lower and upper bounds on the $\Gate{GZZ}$ gate time.
\begin{theorem*}[\cref{thm:gate-time}]
 The optimal total gate time of $\Gate{GZZ}(A)$ with $A \in \HTL (\RR^n)$ is lower and upper bounded by
 \begin{equation}
  \linfnorm{ A \oslash J } \leq \vec 1^T \vec \lambda^* \leq \lonenorm{A \oslash J} \, .
 \end{equation}
\end{theorem*}
Here, the upper bound scales quadratic in $n$ for a constant matrix $M = A \oslash J$.
This upper bound is loose in the sense that it also holds for the $\Gate{GZZ}$ gates constructed by the heuristic in \cref{sec:LP_approx}.
Therefore, we tighten the upper bound to a linear scaling in $n$.
\begin{conjecture*}[\cref{conj:gate-time}]
 The optimal gate time of $\Gate{GZZ}(A)$ with $A \in \HTL (\RR^n)$ is tightly upper bounded by
 \begin{equation}
  \vec 1^T \vec \lambda^* \leq \linfnorm{ A \oslash J } \cdot \begin{cases}
   n \, , &\text{for odd } n\, ,\\
   n-1 \, , &\text{for even } n\, .
  \end{cases}
 \end{equation}
\end{conjecture*}
We provide evidence for this conjecture using an explicit construction (\cref{prop:worst_case_sol}) that realizes the conjectured upper bound for any $n$, as well as numerical evidence for $n\leq 8$ (\cref{fig:worst_case_GZZ}).
Unfortunately, we were not able to prove this result and state the challenges in \cref{app:proof_problems}.

\section{Synthesizing time-optimal GZZ gates is \NP-hard}
\label{sec:LP_is_NP}
In this section, we investigate the complexity of solving the gate synthesis problem stated as \ac{LP}~\eqref{eq:LP1}.
We observe that \ac{LP}~\eqref{eq:LP1} is an optimization over the convex cone generated by
\begin{equation}\label{def:outprod}
	\mathcal E_n \coloneqq \Set*{ \vec m \vec m^T \given \vec m \in \{-1,+1\}^n , m_n = +1} \, ,
\end{equation}
which is the set of outer products generated by all possible encodings $\vec m$.
Due to the symmetry \cref{eq:m_sym} we can uniformly fix the value of one entry of $\vec m$.
We chose the convention $m_n=+1$.
In the literature $\mathcal E_n$ is also known as the \emph{elliptope} of rank one matrices \cite{sdp_cut_polytope}.
In the following, we consider the polytope
\begin{equation}
 \operatorname{conv}(\mathcal E_n) \coloneqq \Set*{ \sum_i \lambda_i \vec r_i \given \sum_i \lambda_i = 1 , \lambda_i \geq 0, \vec r_i \in \mathcal E_n \ } \, ,
\end{equation}
and show the connection to graph theory, in particular to the cuts of graphs.
\begin{definition}[cut polytope \cite{Barahona1986}]
\label{def:cut_polytope}
Let $K_n = (V_n , E_n)$ be a complete graph with $n$ vertices.
Denote $\delta (X)$ the set of all edges with one endpoint in $X \subset V_n$ and the other endpoint in its complement $\bar X$,
i.e., $\delta (X)$ defines the cut between $X$ and $\bar X$.
Let $\chi^{\delta (X)} \in \{0, 1 \}^{|E_n|}$ denote the characteristic vector of a cut, with $\chi_e^{\delta (X)} = 1$ if $e \in \delta (X)$ and $\chi_e^{\delta (X)} =0$ otherwise.
We define the \emph{cut polytope} as the convex hull of the characteristic vectors
\begin{equation}
\operatorname{CUT}_n \coloneqq \operatorname{conv} \Set*{ \chi^{\delta (X)} \in \{0, 1 \}^{|E_n|} \given X \subseteq V_n } \, .
\end{equation}
\end{definition}
\begin{lemma}
\label{lem:cut_ellip}
For all $n$, $\operatorname{CUT}_n$ is isomorphic to $\operatorname{conv}(\mathcal E_n)$.
\end{lemma}
\begin{proof}
For each $X \subset V_n$ we set $m_i = +1$ if $i \in X$ and $m_i = -1$ if $i \in \bar X$.
Note, that there are $2^{n-1}$ different pairs of $X$ and $\bar X$.
We then have $m_i m_j = -1$ if $i \in X$ and $j \in \bar X$ (or the other way around), and $m_i m_j = +1$ if $i, j \in X$ or $i, j \in \bar X$.
So the characteristic vector can be written as $\chi_e^{\delta (X)} = (1-m_i m_j)/2$ for each edge $e \in E_n$ connecting vertices $i$ and $j$.
This is clearly a bijective affine map between the vertices and thus $\operatorname{CUT}_n$ is isomorphic to $\operatorname{conv}(\mathcal E_n)$.
\end{proof}
The following decision problems are membership problems, where the task is to decide if a given element $\vec x$ belongs to a set or not.
In our case $\vec x$ is a vector and the set is a polytope.
\begin{problem}[$\operatorname{CUT}_n$ membership]\label{p:MemCut}\hfill\vspace{-1.4\baselineskip} \\
	\begin{description}[noitemsep,leftmargin=0.5cm,font=\normalfont]
		\item [Instance]
		The adjacency matrix $M \in \HTL(\QQ^n_{\geq 0})$ of a weighted undirected graph with non-negative weights.
		\item [Question]
		Is $M \in \operatorname{CUT}_n$?
	\end{description}
\end{problem}
 \begin{problem}[$\operatorname{conv}(\mathcal E_n)$ membership]\label{p:MemEllip}\hfill\vspace{-1.4\baselineskip} \\
	\begin{description}[noitemsep,leftmargin=0.5cm,font=\normalfont]
		\item [Instance]
		The matrix $M \in \HTL(\QQ^n)$.

		\item [Question]
		Is $M \in \operatorname{conv}(\mathcal E_n)$? That is, does there exist a decomposition $M = \sum_i \lambda_i \vec r_i$ with $\vec r_i \in \mathcal E_n$ such that $\lambda_i \geq 0$ and $\sum_i \lambda_i = 1$?
	\end{description}
\end{problem}
It is well known that the membership problem of the cut polytope, \cref{p:MemCut}, and \cref{p:MemEllip} are \NP-complete \cite{Deza1997a,Garey_2002_Computers,E_Nagy2013}.
Next, we state the decision version of our gate synthesis optimization.
\begin{problem}[time-optimal multi-qubit gate synthesis]\label{p:TOMQs}\hfill\vspace{-1.4\baselineskip} \\
	\begin{description}[noitemsep,leftmargin=0.5cm,font=\normalfont]
		\item [Instance]
		The matrix $M = A \oslash J \in \HTL(\QQ^n)$ and a constant $K \in \QQ_{\geq 0}$.

		\item [Question]
		Is there a decomposition $M = \sum_i \lambda_i \vec r_i$ with $\vec r_i \in \mathcal E_n$ such that $\lambda_i \geq 0$ and $\sum_i \lambda_i \leq K$?
	\end{description}
\end{problem}

\begin{theorem}\label{thm:hardness}
\cref{p:TOMQs}, which is the decision version of the \ac{LP}~\eqref{eq:LP1}, is \NP-complete.
\end{theorem}
\begin{proof}
A solution of \cref{p:TOMQs} can be verified in polynomial time since there always exists a decomposition $\sum_i \lambda_i \vec r_i$ of $M$ with minimal $\sum_i \lambda_i$ which has at most $\binom{n}{2} = n/2(n-1)$ non-zero terms \cite[Proposition~II.3]{basler_synthesis_2022}.
Therefore, \cref{p:TOMQs} is \NP.

To show that \cref{p:TOMQs} is \NP-hard, we construct a polynomial-time mapping reduction from \cref{p:MemEllip} to \cref{p:TOMQs}.
Given the matrix $M \in \HTL(\QQ^n)$ and a constant $K \in \QQ_{\geq 0}$ as an instance for \cref{p:TOMQs}, let $\lambda_i \geq 0$ be the positive coefficients of the decomposition.
If we find $\sum_i \lambda_i < K$, then we can always add additional $\lambda$'s such that equality holds.
We choose the additional $\lambda$'s as the coefficients of the decomposition for the all-zero matrix, see e.g.\ \cref{thrm:tensor_GZZ} below with $\vec k = \vec 1_n$ for an explicit construction.
We define the matrix $M^\prime \coloneqq M/K$ and the positive coefficients $\lambda^\prime_i \coloneqq \lambda_i/K$.
Then $M^\prime = \sum_i \lambda^\prime_i \vec r_i$ with $\vec r_i \in \mathcal E_n$ and $\sum_i \lambda^\prime_i = 1$.
\end{proof}

\section{Time-optimal GZZ synthesis for special instances}
\label{sec:ana_GZZ}
Although, solving the \ac{LP}~\eqref{eq:LP1} is \NP-hard we present explicit optimal solutions for certain families of instances which is equivalent to constructing time-optimal $\Gate{GZZ}$ gates.
The constructions of this section yield a qubit-independent total $\Gate{GZZ}$ time which satisfies the optimal lower bound of \cref{lem:triv_low_bound} below.
Moreover, we show that some $\Gate{GZZ}$ gates can be synthesized with an encoding cost independent of the number of qubits.
However, most of these constructions assume constant values for the elements of the band-diagonal of the physical coupling matrix $J$.
These assumptions are relaxed throughout this section providing explicit optimal solutions for physically relevant cases.
These results build the foundation for the heuristic algorithm for fast $\Gate{GZZ}$ gate synthesis in the next section.

By $\lpnorm{M}$ we denote the $\ell_p$-norm of a symmetric matrix $M$, which is given as the $\ell_p$-norm of a vector $\vec v(M)$ containing all lower/upper triangular matrix elements in some order.
First, we proof a lower bound on the optimal total $\Gate{GZZ}$ time which can be used to verify time-optimality.
\begin{lemma}
\label{lem:triv_low_bound}
For any $M \in \HTL(\RR^n)$ the optimal objective function value of the \ac{LP}~\eqref{eq:LP1} is lower bounded by
 \begin{equation}
 \linfnorm{ M } \leq \vec 1^T \vec \lambda^* \, .
 \end{equation}
\end{lemma}
\begin{proof}
The lower bound can be verified by the fact that the matrix representation $V$ of the linear operator in \cref{eq:lin_op} only has entries $\pm 1$ and that $\vec \lambda^*$ is non-negative.
Thus, it holds that $\linfnorm{ M } = \linfnorm{ V \vec \lambda^* } \leq \vec 1^T \vec \lambda^*$.
\end{proof}
Next, we provide calculation rules for the coupling matrix $A\in \HTL(\RR^n)$ of the $\Gate{GZZ}(A)$ gate. These rules are inherited from matrix exponentials.
Let $A_1, A_2, A_3 \in \HTL(\RR^n)$ then
\begin{enumerate}[label=(\roman*)]\itemsep=0pt
	\item\label{item:sum_GZZ} 
        $\Gate{GZZ} (A_1 + A_2) = \Gate{GZZ} (A_1)\Gate{GZZ} (A_2)$, 
	\item\label{item:direct_sum_GZZ} 
        $\Gate{GZZ} (A_1 \oplus A_2) = \Gate{GZZ} (A_1) \otimes \Gate{GZZ} (A_2)$ and 
	\item\label{item:had_prod_GZZ} 
        $\Gate{GZZ} (A_1 \circ A_2) = \Gate{GZZ} (A_3)$, where the coupling matrices can be decomposed as
	    \begin{equation*}
		    A_1 \coloneqq \sum_{k=1}^{d_1} \lambda_{\vec m_k} \vec m_k \vec m_k^T, \quad A_2 \coloneqq \sum_{k=1}^{d_2} \beta_{\vec v_k} \vec v_k \vec v_k^T, \quad A_3 \coloneqq \sum_{k=1}^{d_3} \gamma_{\vec w_k} \vec w_k \vec w_k^T \, ,
	    \end{equation*}
	    with $d_3 = d_1 d_2$, $\vec w_k = ( \vec m_i \circ \vec v_j)_{k = d_2 i+j}$ and $\gamma_{\vec w_k} = (\lambda_{\vec m_i} \beta_{\vec v_j})_{k = d_2 i+j}$ for $j = 1, \dots , d_2$ and $i = 0, \dots , d_1-1$.
\end{enumerate}
In \cref{item:had_prod_GZZ} we can also express $\vec \gamma = \vec \lambda \otimes \vec \beta$ with the Kronecker product.
Note, that these generated $\Gate{GZZ}$ gates are not necessarily time-optimal.

We denote the $k \times k$ identity matrix by $\1_{k}$, the $k$ dimensional all-ones vector by $\vec 1_k$ and the $k \times k$ matrix of ones with vanishing diagonal by $E_k \coloneqq \vec1_{k} \vec1_{k}^T - \1_k$.
With the next two lemmas, we bound the encoding cost and total gate time for special cases of \cref{item:sum_GZZ} and \cref{item:direct_sum_GZZ}, where the total gate time is constant.
These results will provide a basis for all other constructions.
The following result constructs total coupling matrices $A \in \HTL (\RR^n)$ on arbitrary subsets of qubits.
\begin{lemma}
\label{thrm:tensor_GZZ}
 Let the coupling matrix $J$ be constant, i.e.\ $J_{ij} = c \in \RR_{\geq 0}$ for all $i,j \in [n]$ with $i\neq j$.
 For any $\varphi \in \RR$ and $s \in \ZZ_{\geq 0}$ the $\Gate{GZZ}(A)$ with the matrix
\begin{equation}\label{eq:tensor_GZZ}
 A = \varphi \bigoplus_{i=1}^s E_{k_i} \, ,
\end{equation}
 with $k_i \geq 1$ for $i=1, \dots , s$ has the encoding cost of $d = 2^{\lceil \log_2 (s) \rceil} \leq 2s$ and constant total $\Gate{GZZ}$ time $\varphi / c$.
 This total gate time is optimal.
\end{lemma}
\begin{proof}
W.l.o.g.\ we set $c = \varphi = 1$.
We have $n \coloneqq \sum_{i=1}^s k_i$ qubits.
Note, that $E_{k_i} = E_{1} = 0$ only contributes to an entry in the diagonal of $A$, and hence this qubit does not participate in the $\Gate{GZZ}(A)$ gate.
We denote the $d \times d$ Hadamard matrix by $H^{d \times d}$ and the matrix consisting of its first $s$ columns by $H^{d \times s}$ where $s \leq n$.
The orthogonality of the columns of any Hadamard matrix yields $(H^{d \times s})^T H^{d \times s} = d\, \1_s$.
Replacing each $i$-th column by $k_i$ copies of it we obtain the $d\times n$-matrix $H^{d \times n}_{\vec k}$ from $H^{d \times s}$.
Then, we attain the diagonal block matrix structure
\begin{equation}
\label{eq:diag_block_struct}
 (H^{d \times n}_{\vec k})^T H^{d \times n}_{\vec k} - d \1_{n}  = d \bigoplus_{i=1}^s E_{k_i} \, .
\end{equation}
We set the diagonal elements on the left-hand side to zero, since they only contribute to an energy offset in the Hamiltonian.
Take each row of $H^{d \times n}_{\vec k}$ as a possible vector $\vec m\in \{-1,1\}^n$
to construct the total coupling matrix
\begin{equation}
\label{eq:tensor_GZZ_2}
 A = \frac{1}{d} \sum_{\vec m\in \operatorname{rows}(H_{\vec k}^{d \times n})}\vec m \vec m^T \, ,
\end{equation}
 i.e., the time steps have been chosen as $\lambda_{\vec m} = \frac 1 d \, [\vec m \in \operatorname{rows}(H^{d \times n}_{\vec k})]$, with the Iverson bracket $[\, \cdot \,]$.
Clearly, we have $\sum_{\vec m} \lambda_{\vec m} = 1$.
If we take the constants $c$ and $\varphi$ into account, we just multiply \cref{eq:tensor_GZZ_2} by $\varphi / c$ which gives a total $\Gate{GZZ}$ time of $\varphi / c$.
Furthermore, this total $\Gate{GZZ}$ time is optimal since $\linfnorm{ M} = \linfnorm{ A \oslash J} = \varphi / c$ satisfies the lower bound of \cref{lem:triv_low_bound}.
\end{proof}
The encoding cost of $\Gate{GZZ} (A)$ considered in this lemma can be reduced if redundant encodings $\vec m = \vec m^\prime \in \{-1,1\}^n$ are present by adding the corresponding time steps $\lambda_{\vec m} + \lambda_{\vec m^\prime}$.
It can be further reduced by using the Hadamard conjecture \cite{paleyn_1933,hedayat_hadamard_1978}.
If the Hadamard conjecture holds, then there exists Hadamard matrices of any dimension divisible by $4$.
It is known that the Hadamard conjecture holds for dimensions $s \leq 668$ \cite{kharaghani_2005,dokovic_2014}, thus the encoding cost can be reduced to $d = 4 \lceil s/4 \rceil \leq s-4$ in this regime.
The encoding cost of the following \cref{thm:apdx_NNcoupling}, \cref{cor:apdx_subdiagonal,cor:apdx_excl_qubit} and the efficient heuristic in \cref{sec:LP_approx} can be reduced in the same fashion.

The assumption in \cref{thrm:tensor_GZZ} of a constant coupling matrix, $J_{i\neq j} = c$, is physically unreasonable.
Therefore, we relax this assumption for block sizes $k_i=2$ which corresponds to pairwise next neighbor couplings.
We use \cref{thrm:tensor_GZZ} to construct time-optimal $\Gate{GZZ}$ gates for a certain family of coupling matrices.

\begin{lemma}\label{thm:apdx_NNcoupling}
 Let $n$ be even and the coupling matrix $J$ be constant on the first subdiagonal (the other elements are arbitrary),
 i.e.\ $J_{ij} = c$ for $i\in [n-1]$ and $j = i\pm 1$.
 Then
 \begin{equation}
  \Gate{GZZ}(A) = \bigotimes_{i=1}^{n/2} \Gate{GZZ}(\varphi E_{2} )  \, ,
 \end{equation}
  has the encoding cost of $d = 2^{\lceil \log_2 (n)-1 \rceil} < n$ and constant total $\Gate{GZZ}$ time $\varphi / c$. 
  This total gate time is optimal.
\end{lemma}

\begin{proof}
Since the matrix $E_{2}$ only has non-zero entries in the first subdiagonal, the coupling matrix $J$ only needs to be constant there.
The claim follows immediately from \cref{thrm:tensor_GZZ} by setting $k_1=\dots = k_{n/2} = 2$.
Clearly, the total $\Gate{GZZ}$ time is optimal since it saturates the lower bound of \cref{lem:triv_low_bound}.
\end{proof}

\cref{thm:apdx_NNcoupling} guarantees an encoding cost of $d < n$, which is a quadratic saving compared to the general \ac{LP} solution with encoding cost $\binom{n}{2}$.
We note, that
\begin{equation}
 A = \varphi \bigoplus_{i=1}^{n/2} \begin{pmatrix} 0&1\\1&0\end{pmatrix}
\end{equation}
corresponds to parallel $\Gate{ZZ}(\varphi)$ gates, which find applications in simulating molecular dynamics \cite{Cohn21QuantumFilter}.
The assumption of a constant subdiagonal of $J$ can be realized in an ion trap platform by applying an anharmonic trapping potential \cite{johanning_isospaced_2016}.

By combining \ref{item:sum_GZZ}, \cref{thrm:tensor_GZZ} and \cref{thm:apdx_NNcoupling} we obtain the $\Gate{GZZ}$ gate for next neighbor coupling.
\begin{theorem}
\label{cor:apdx_subdiagonal}
 Let the subdiagonal of $J$ be constant, i.e.\ $J_{ij} = c$ for $i\in [n-1]$ and $j = i\pm 1$.
 Then $\Gate{GZZ}(A)$ on $n$ qubits with
 \begin{equation}
  A = \varphi \begin{pmatrix}
		0&1&&&&&\\
		1&0&1&&&&\\
		&1&0&1&&&\\
		&&1&0&1&&\\
		&&&&\ddots&&\\
	\end{pmatrix}  \, ,
 \end{equation}
  has the encoding cost of $d \leq 2n$ (for $n>4$) and constant total $\Gate{GZZ}$ time $2 \varphi / c$. 
  This total gate time is optimal.
\end{theorem}
\begin{proof}
We set $c = \varphi = 1$ w.l.o.g.
For now, we assume that $n$ is even.
Then
\begin{equation}
\begin{aligned}
 A &= \bigoplus_{i=1}^{n/2} E_{2} + E_{1} \bigoplus \left( \bigoplus_{i=1}^{n/2-1} E_{2} \right) \bigoplus E_{1} \\
 &= \begin{pmatrix}
		0&1&&&&&\\
		1&0&&&&&\\
		&&0&1&&&\\
		&&1&0&&&\\
		&&&&\ddots&&\\
	\end{pmatrix} + \begin{pmatrix}
		0&&&&&&\\
		&0&1&&&&\\
		&1&0&&&&\\
		&&&0&1&&\\
		&&&1&0&&\\
		&&&&&\ddots&\\
	\end{pmatrix} \\
&= \begin{pmatrix}
		0&1&&&&&\\
		1&0&1&&&&\\
		&1&0&1&&&\\
		&&1&0&1&&\\
		&&&&\ddots&&\\
	\end{pmatrix}
\end{aligned}
\end{equation}
Again, the diagonal entries do not contribute to the interactions.
The first term can be implemented, using \cref{thm:apdx_NNcoupling}, with the encoding cost $d_1 = 2^{\lceil \log_2 (n)-1 \rceil}$ and the $\Gate{GZZ}$ time of $1$.
The second term can be implemented, using \cref{thrm:tensor_GZZ}, as
\begin{equation}
 \bigotimes_{i=1}^{s} \Gate{GZZ}(\varphi E_{k_i}) \, ,
\end{equation}
where $s = n/2+1$,  $k_1 = k_{s} = 1$ and $k_i = 2$ for $i = 2, \dots , n/2$, with the encoding cost $d_2 = 2^{\lceil \log_2 (n+2)-1 \rceil}$ and the $\Gate{GZZ}$ time of $1$.
Adding the encoding costs $d = d_1 + d_2$ and $\Gate{GZZ}$ times of both terms yields the desired result.
If $n$ is not even, then repeat the previous steps for $n+1$ but in the end reduce the dimension of all the resulting $\vec m$ by discarding the last entry.

This construction corresponds to a feasible solution of the primal \ac{LP}~\eqref{eq:LP1} with the objective function value $2$.
To show optimality it suffices to construct a feasible solution for the dual \ac{LP}~\eqref{eq:dualLP} with the objective function value of $2$.
First, consider the case $n=3$ with the total coupling matrix
\begin{equation}
   A = \begin{pmatrix}
		0&1&0\\
		1&0&1\\
		0&1&0\\
	\end{pmatrix}  \Rightarrow \vec v(A)^T = (1, 0, 1) \, ,
\end{equation}
and the matrix representation of the linear operator
\begin{equation}\label{eq:V_n3}
   V^T = \begin{pmatrix}
		1&1&1\\
		1&-1&-1\\
		-1&1&-1\\
		-1&-1&1\\
	\end{pmatrix} \, ,
\end{equation}
as in \cref{eq:lin_op}.
A feasible dual solution satisfying $V^T \vec y \leq \vec 1$ is $\vec y = (1, -1, 1)^T$.
Thus, we can verify optimality for $n=3$ since the objective function value is $\innerp{A}{\vec y} = \vec v(A)^T \vec y = 2$.
Now, we consider arbitrary $n>3$.
Extending the dual solution for the case $n=3$ with zeros $\vec y = (1, -1, 1, 0, \dots , 0)^T$ does not change the objective function value $\innerp{A}{\vec y} = \vec v(A)^T \vec y = 2$.
Such an extended dual solution is still feasible since $V^T$ restricted to the first three columns only has rows which are already contained in \cref{eq:V_n3} due to the symmetry of \cref{eq:m_sym}.
\end{proof}
\cref{alg:apdx_subdiagonal} is the pseudocode implementation of \cref{cor:apdx_subdiagonal}.
It takes the number of qubits $n$, the constant value $c$ of the subdiagonal of $J$ and the factor $\varphi$ of $A$ as input and returns the sparse vector $\vec \lambda$, containing the time steps.
The encodings $\vec m$ are given by the indices of the non-zero elements $\lambda_{\vec m} \neq 0$.
\begin{algorithm}[H]
	\caption{Synthesize $\Gate{GZZ}(A)$ as in \cref{cor:apdx_subdiagonal}.}\label{alg:apdx_subdiagonal}
	\raggedright\textbf{Input:} $n, c, \varphi$
	\begin{algorithmic}[0]
		\State Initialize $\vec \lambda = \vec 0 \in \RR^{2^{n}}$
		\State{is\_odd $ \gets false$}
		\If{$n$ odd}
			\State{$n \gets n+1$}
			\State{is\_odd $ \gets true$}
		\EndIf
		\State Let $H_1^{d_1 \times d_1}$ be a Hadamard matrix \Comment{With $d_1 = 2^{\lceil \log_2 (n)-1 \rceil}$ as in \cref{thm:apdx_NNcoupling}}
		\State{$H_1^{d_1 \times \frac{n}{2}} \gets \frac{n}{2} \text{ columns of } H_1^{d_1 \times d_1}$} \Comment{It does not matter which columns}
		\State{$H_1^{d_1 \times n} \gets \text{duplicate columns } i=1, \dots , \frac{n}{2} \text{ of } H_1^{d_1 \times \frac{n}{2}}$}
		\State Let $H_2^{d_2 \times d_2}$ be a Hadamard matrix \Comment{With $d_2 = 2^{\lceil \log_2 (n+2)-1 \rceil}$}
		\State{$H_2^{d_2 \times (\frac{n}{2}+1)} \gets \frac{n}{2}+1 \text{ columns of } H_2^{d_2 \times d_2}$} \Comment{It does not matter which columns}
		\State{$H_2^{d_2 \times n} \gets \text{duplicate columns } i=2, \dots , \frac{n}{2} \text{ of } H_2^{d_2 \times (\frac{n}{2}+1)}$}
		\For{$j \in \{1, 2\}$}
		\If{is\_odd}
			\State Delete one column of $H_j^{d_j \times n}$  \Comment{It does not matter which column}
		\EndIf
		\For{$\vec m \in \operatorname{rows}(H_j^{d_j \times n})$}
			\State{$\lambda_{\vec m} \gets \frac{\varphi}{d_j c}$}
		\EndFor
		\EndFor
	\end{algorithmic}
	\raggedright\textbf{Output:} $\vec \lambda$ \Comment{Can efficiently be saved in a sparse data format}
\end{algorithm}

The following theorem does not require any additional assumptions on $J$.
It shows, how the \ac{LP}~\eqref{eq:LP1} can be supplemented with the explicit solution to exclude certain qubits.
\begin{theorem}[Excluding qubits]
\label{cor:apdx_excl_qubit}
Let $N$ be the total number of qubits on the quantum hardware and $n = N-s$ be the participating qubits in the $\Gate{GZZ}$ gate.
Synthesize the $\Gate{GZZ}(A)$ gate with $A \in \HTL (\RR^{n})$, using the \ac{LP}~\eqref{eq:LP1}.
Then, the total encoding cost (on $N$ qubits) is at most $\binom{n}{2} 2^{\lceil \log_2 (s+1) \rceil} $ and the total $\Gate{GZZ}$ time, $\vec1^T \vec \lambda^*$, (on $N$ qubits) is the same as for the \ac{LP} run on $n$ qubits.
\end{theorem}
\begin{proof}
 Assume w.l.o.g. that all qubits to be excluded are at the end of the qubit array.
Let $k_1 = n$ and $k_2 , \dots , k_{s+1} = 1$.
Using \cref{thrm:tensor_GZZ} we obtain an
encoding cost of $d_1 = 2^{\lceil \log_2 (s+1)\rceil}$ and total $\Gate{GZZ}$ time $1$ to generate the matrix
\begin{equation}
 A_1 = \begin{pmatrix}
		E_n&0\\
		0&0\\
	\end{pmatrix} \, .
\end{equation}
Solving the \ac{LP}~\eqref{eq:LP1} for a matrix $A \in \HTL (\RR^{n})$ yields the encoding cost $d_2 = \binom{n}{2}$ and the total $\Gate{GZZ}$ time $\vec1^T \vec \lambda^*$.
We define the extension of $A \in \HTL (\RR^{n})$ by $A_2 \in \HTL (\RR^{N})$.
The extension can be done by appending $s$ arbitrary elements in $\{ -1 , +1 \}$ to all vectors $\vec m \in \{ -1 , +1 \}^n$ given by the \ac{LP}~\eqref{eq:LP1}.
Clearly,
\begin{equation}
 A_1 \circ A_2 = \begin{pmatrix}
		E_n&0\\
		0&0\\
	\end{pmatrix} \circ \begin{pmatrix}
		A&*\\
		*&*\\
	\end{pmatrix} = \begin{pmatrix}
		A&0\\
		0&0\\
	\end{pmatrix} \, .
\end{equation}
By \ref{item:had_prod_GZZ}, the total encoding cost is $d_1 d_2$ and the total $\Gate{GZZ}$ time is
\begin{equation}
\sum_{i = 1}^{d_1} \sum_{j = 1}^{d_2} 1/d_1 \lambda_j^* = \vec1^T \vec \lambda^* \, .
\end{equation}
Consider now an arbitrary coupling matrix $J$.
Then
\begin{equation}
 J \circ A_1 \circ A_2 =J \circ \begin{pmatrix}
		E_n&0\\
		0&0\\
	\end{pmatrix} \circ \begin{pmatrix}
		\tilde A&*\\
		*&*\\
	\end{pmatrix} = \begin{pmatrix}
		A&0\\
		0&0\\
	\end{pmatrix} \, ,
\end{equation}
where $\tilde A = A \oslash J$ is decomposed by the \ac{LP}~\eqref{eq:LP1}.
\end{proof}
\cref{alg:apdx_excl_qubit} takes the total number of qubits $N$, the total coupling matrix $A$ (on $n$ qubits) and the set $\mathcal Z \coloneqq \Set*{i \in [N] \given \text{exclude qubit } i }$ as input and returns the sparse vector $\vec \gamma$, containing the time steps. The encodings $\vec w$ are given by the indices of the non-zero elements $\gamma_{\vec w} \neq 0$.
\begin{algorithm}[H]
	\caption{Excluding qubits.}\label{alg:apdx_excl_qubit}
	\raggedright\textbf{Input:} $N, A, \mathcal Z$ (set of qubit indices to be excluded)
	\begin{algorithmic}[0]
		\State Initialize $\vec \gamma = \vec 0 \in \RR^{2^{N}}$
		\State{$s \gets |\mathcal Z|$}
		\State{$\overline{\mathcal Z} \gets \Set*{i \in [N] \given i \notin \mathcal Z}$}
		\State Let $H_1^{d_1 \times d_1}$ be a Hadamard matrix \Comment{With $d_1 = 2^{\lceil \log_2 (s+1) \rceil}$}
		\State{$H_1^{d_1 \times (s+1)} \gets s+1 \text{ columns of } H_1^{d_1 \times d_1}$} \Comment{It does not matter which columns}
		\State{$H_1^{d_1 \times N} \gets$ duplicate one column of $H_1^{d_1 \times (s+1)}$ $n-1$ times and place them at indices $i \in \overline{\mathcal Z}$}
		\State{$\vec \lambda^* \gets \text{Solve } A = \sum_{\vec m} \lambda_{\vec m} \vec m \vec m^T$} \Comment{using \ac{LP}~\eqref{eq:LP1}}

		\For{$\vec v \in \operatorname{rows}(H_1^{d_1 \times N})$}
		\For{$\vec m \in \Set*{\vec m \given \lambda_{\vec m}^* \neq 0}$} \Comment{There are at most $\binom{n}{2}$ such $\vec m$}
			\State{$\tilde{\vec m} \gets $ extend $\vec m$ with arbitrary elements from $\{-1, +1\}$ at indices $i \in \mathcal Z$.}
			\State{$\vec w \gets \vec v \circ \tilde{\vec m}$}
			\State{$\gamma_{\vec w} \gets \frac{\lambda_{\vec m}^*}{d_1}$}
		\EndFor
		\EndFor
	\end{algorithmic}
	\raggedright\textbf{Output:} $\vec \gamma$ \Comment{Can efficiently be saved in a sparse data format}
\end{algorithm}

We showed explicit constructions of time-optimal $\Gate{GZZ}$ gates for total coupling matrices $A \in \HTL (\RR^n)$ with diagonal block structure and next neighbor couplings.
The resulting $\Gate{GZZ}$ gates have a constant gate time and require only linear many encodings to be implemented.

\section{Efficient heuristic for fast GZZ gates}
\label{sec:LP_approx}
In this section, we build on the results of \cref{thrm:tensor_GZZ} to derive a heuristic algorithm for synthesizing $\Gate{GZZ}(A)$ gates with low total gate time for any $A \in \HTL (\RR^n)$.
This algorithm runs in polynomial time as opposed to the general \ac{LP}~\eqref{eq:LP1}, which we have shown in \cref{thm:hardness} to be \NP-hard.
The runtime saving is due to the restriction of the elliptope $\mathcal E_n$ in \cref{def:outprod}, with exponential many elements, to a set with polynomial many elements.
This restriction yields a polynomial sized \ac{LP} which can be solved in polynomial time.
In practice, the simplex algorithm has a runtime that scales polynomial in the problem size \cite{spielman_smoothed_2004}.

Recall the modified Hadamard matrix $H^{d \times n}_{\vec k}$ defined in the proof of \cref{thrm:tensor_GZZ}, where we used the rows of $H^{d \times n}_{\vec k}$ as encodings $\vec m$ to generate block diagonal target coupling matrices under some assumptions.
Here, $s$ is the number of block matrices on the diagonal of the target matrix, $\vec k \in \NN^s$ contains the dimensions for each block and $d = 2^{\lceil \log_2 (s) \rceil}$ is the required number of encodings to construct such a block diagonal matrix.
From now on, we only consider $\vec k = (j, 1 \dots 1) \in \NN^{s}$ such that
\begin{equation}
\label{eq:heuristic_outer_prod}
 (H^{d \times n}_{\vec k})^T H^{d \times n}_{\vec k}  - d \1_{n} = d \left(E_{j} \oplus E_{1} \oplus \cdots \oplus E_{1}\right) \, ,
\end{equation}
see \cref{eq:diag_block_struct}. 
The requirement that $\sum_i k_i = n$ implies that such a vector $\vec k$ has $s = n-j+1$ entries.
Permuting the columns of $H^{d \times n}_{\vec k}$ results in the same permutation of the rows and columns of the right-hand side of \cref{eq:heuristic_outer_prod}.
We denote the set of all column-permuted $H^{d \times n}_{\vec k}$ by $\mathcal C^{(j)}$.
A specific element of $\mathcal C^{(j)}$ is denoted by $H^{d \times n}_{\vec r}$, where $\vec r \in \NN^j$ is an ordered multi-index $r_1 < \dots < r_j$ indicating which columns of $H^{d \times n}_{\vec r}$ are identical. 
For example, $\vec r = (2,5,6)$ indicates that the columns of $H^{d \times n}_{\vec r}$ with indices $2$, $5$ and $6$ are identical, i.e.\ replacing column $5$ and $6$ with column $2$.
This notation will be useful later.

\begin{definition}
For any $j \in \{2,3,\dots,n\}$, we define the restricted elliptope
\begin{equation}\label{def:restricted_outprod}
	\mathcal E_n^{(j)} \coloneqq \Set*{ \vec m \vec m^T \given \vec m \text{ is a row of } H^{d \times n}_{\vec r} \in \mathcal C^{(j)} } .
\end{equation}
Further we define
\begin{equation}
 \mathcal E_n^{[j]} \coloneqq \bigcup_{i=2}^{j} \mathcal E_n^{(i)} \, .
\end{equation}
\end{definition}
We choose the definition in \cref{def:restricted_outprod} similar as in \cref{def:outprod}.
Next, we show the size scaling of the restricted elliptopes.
This directly translates to the size and runtime of the heuristic synthesis optimization.

\begin{proposition}
\label{prop:poly_size_ellip}
 For any $j \in \{2,3,\dots,n\}$, the number of different encodings $\vec m$ generating the restricted elliptope $\mathcal E_n^{[j]}$ scales as $\LandauO (n^{j+1})$.
\end{proposition}
\begin{proof}
Note, that $|\mathcal C^{(j)}| = \binom{n}{j}$ since there are $j$ duplicate columns in $H^{d \times n}_{\vec r}$.
The binomial coefficient can be bounded by $\binom{n}{j} \leq n^j/j!$.
Since there are $d = 2^{\lceil \log_2 (n-j+1) \rceil} < 2(n-j+1)$ rows of $H^{d \times n}_{\vec r}$ we have a rough upper bound of the number different encodings generating the restricted elliptope, $\abs[\big]{\mathcal E_n^{(j)}} \leq d \binom{n}{j} < 2(n-j+1) \binom{n}{j} < 2 n^{j+1}/j!$.
The first inequality is due to possible redundant encodings in the definition of $\mathcal E_n^{(j)}$.
Similarly, we can upper bound $\abs[\big]{\mathcal E_n^{[j]}} \leq \sum_{i=2}^j \abs[\big]{\mathcal E_n^{(i)}} < 2 \sum_{i=2}^j n^{i+1}/i!$ which is a polynomial of order $j+1$.
\end{proof}
We denote the convex cone generated by a set $V$ by
\begin{equation}
 \operatorname{cone} (V) \coloneqq \Set*{ \sum_{i} \lambda_i v_i \given \lambda_i \geq 0 \text{, } v_i \in V } \, .
\end{equation}

With that, we are ready to present the main result of this section.
\begin{theorem}
\label{thrm:spanning_set}
$\operatorname{cone} (\mathcal E_n^{(2)}) = \HTL (\RR^{n})$.
\end{theorem}
\begin{proof}
W.l.o.g.\ we can assume $d=n$ and denote $H^{d \times d}_{\vec r} \in \mathcal C^{(2)}$ by $H^{d}_{(r_1, r_2)}$ with the property
\begin{equation}
\label{eq:std_basis_decomp}
 (H^{d}_{(r_1, r_2)})^T H^{d}_{(r_1, r_2)} - d \1_{d} = d\, e_{(r_1, r_2)} \, ,
\end{equation}
where $e_{(r_1, r_2)}$ is an element of the standard basis for symmetric matrices with vanishing diagonal.
By \cref{eq:std_basis_decomp} it holds $\HTL (\RR_{\geq 0}^{n}) \subseteq \operatorname{cone} (\mathcal E_n^{(2)})$, i.e., symmetric matrices with non-negative entries are in the convex cone.

It is left to show that $\HTL (\RR_{< 0}^{n}) \subseteq \operatorname{cone} (\mathcal E_n^{(2)})$, i.e., that also symmetric matrices with negative entries are in the convex cone.
To show this inclusion we define $H^{d}_{(r_1, -r_2)}$ similar as $H^{d}_{(r_1, r_2)}$ except the duplicate column at $r_2$ is multiplied by $-1$ such that
\begin{equation}
\label{eq:minus_std_basis_decomp}
 (H^{d}_{(r_1, -r_2)})^T H^{d}_{(r_1, -r_2)} - d \1_{d} = -d\, e_{(r_1, r_2)} \, .
\end{equation}
We have to show that for each row $\vec m \in \operatorname{rows} (H^{d}_{(r_1, -r_2)})$ there exist $\tilde r_1$ and $\tilde r_2$ such that $\vec m \in \operatorname{rows} (H^{d}_{(\tilde r_1, \tilde r_2)})$.
This can be verified straightforwardly for $d = 4$ by checking all rows.
W.l.o.g.\ we show that the hypothesis holds for any $H^{2d}_{(r_1, -r_2)}$ by assuming it holds for $H^{d}_{(r_1, -r_2)}$.
The Sylvester-Hadamard matrix is constructed inductively according to
\begin{equation}
 H^{2d} = \begin{pmatrix}
		H^d & H^d\\
		H^d & -H^d\\
	\end{pmatrix} .
\end{equation}
We consider three cases for $H^{2d}_{(r_1, -r_2)}$.
\paragraph*{Case 1.}
For a $H^{2d}_{(r_1, -r_2)}$ with identical columns, up to minus sign, at indices $r_1, r_2 \in [d]$ or $r_1, r_2 \in [d+1, 2d]$ the hypothesis holds by our assumption by choosing $\tilde r_1, \tilde r_2 \in [d]$ or $\tilde r_1, \tilde r_2 \in [d+1, 2d]$ respectively.
\paragraph*{Case 2.}
Considering the first $d$ rows of $H^{2d}_{(r_1, -r_2)}$ with identical columns, up to minus sign, at indices $r_1 \in [d]$ and $r_2 \in [d+1, 2d]$.
This case is equivalent to \textit{Case 1.} with $r_1, r_2 \in [d+1, 2d]$ since only the column at $r_2$ is negated.
\paragraph*{Case 3.}
Considering the last $d$ rows of $H^{2d}_{(r_1, -r_2)}$ with identical columns, up to minus sign, at indices $r_1 \in [d]$ and $r_2 \in [d+1, 2d]$.
These rows are included in the last $d$ rows of $H^{2d}_{(\tilde r_1, \tilde r_2)}$ with $\tilde  r_1, \tilde  r_2 \in [d+1, 2d]$ and $\tilde  r_2 = r_2$ since the column $r_2$ is negated which is equivalent to just duplicating a column of $-H^d$.

We have shown that for each row $\vec m \in \operatorname{rows} (H^{d}_{(r_1, -r_2)})$ there exist $\tilde r_1$ and $\tilde r_2$ such that $\vec m \in \operatorname{rows} (H^{d}_{(\tilde r_1, \tilde r_2)})$.
Therefore, $\HTL (\RR_{< 0}^{n}) \cup \HTL (\RR_{\geq 0}^{n}) = \HTL (\RR^{n}) = \operatorname{cone} (\mathcal E_n^{(2)})$.
The last equality follows from the definition of $\operatorname{cone}$ and $\mathcal E_n^{(2)}$.
\end{proof}
\Cref{thrm:spanning_set} shows that the constraint of \ac{LP}~\eqref{eq:LP1} can always be fulfilled only considering $\vec m \vec m^T \in \mathcal E_n^{(2)}$.
Similar to \cref{eq:lin_op} we define the restricted linear operator $\mathcal V^{[j]}: \RR_{\geq 0}^{s} \rightarrow \HTL(\RR^n) : \vec \lambda \mapsto \sum_{\vec m} \lambda_{\vec m} \vec m \vec m^T$ for all $\vec m \vec m^T \in \mathcal E_n^{[j]}$ with $h = |\mathcal E_n^{[j]}|$, represented by a matrix $V^{[j]}  \in \{-1,+1 \}^{(\binom{n}{2}) \times h}$.
We define the \emph{restricted \ac{LP}$^j$} to be
\begin{mini}
{}{\vec 1^T \vec \lambda}{}{}
\label{eq:LP_poly_hierarchy}
\addConstraint{V^{[j]} \vec \lambda} {=\vec v(M)}{}{}
\addConstraint{ \vec \lambda}{ \in \R^{h}_{\geq 0} \, .}{}{}
\end{mini}
\cref{alg:eff_heu} summarizes the steps to construct $\mathcal E_n^{[j]}$ and therefore the matrix representation of the restricted linear operator $V^{[j]}$.
In practice, \cref{alg:eff_heu} has to be executed only once per number of qubits $n$ since the constraints of \ac{LP}~\eqref{eq:LP_poly_hierarchy} can be fulfilled for all $M \in \HTL (\RR^{n})$.
This is due to the fact that $\mathcal E_n^{(2)} \subseteq \mathcal E_n^{[j]}$ for any $2 \leq j \leq n$ and $\operatorname{cone} (\mathcal E_n^{(2)}) = \HTL (\RR^{n})$ (\cref{thrm:spanning_set}).
The time and space complexity of \cref{alg:eff_heu} scales polynomially in $n$ as shown in \cref{prop:poly_size_ellip}.
Therefore, the restricted \ac{LP}$^j$ is also of polynomial size for a fixed $j$.
Increasing $j$ leads to better approximations due to the enlarged search space for the optimal solution.
Note, that the runtime of the \ac{MIP} defined in \cite[Section 2.2.2]{basler_synthesis_2022} also benefits from using $\mathcal E_n^{[j]}$.

As mentioned in \cref{sec:ana_GZZ} the dimension of the Hadamard matrices in \cref{eq:heuristic_outer_prod} can be reduced to $d = 4 \lceil (n-1)/4 \rceil \leq n-5$ if we assume that the Hadamard conjecture holds.
Therefore, the runtime of the restricted \ac{LP}$^j$ is reduced as well if such Hadamard matrices are used.
In \cref{sec:numerics} we numerically benchmark the heuristic algorithm with and without the reduced runtime of the restricted \ac{LP}$^j$.

\begin{algorithm}[H]
	\caption{Constructing $\mathcal E_n^{[j]}$.}\label{alg:eff_heu}
	\raggedright\textbf{Input:} $n, j$
	\begin{algorithmic}[0]
		\State Initialize $\mathcal E_n^{[j]} = \emptyset$
		\For{$i = 2, \dots, j$}
		\State Initialize $\mathcal E_n^{(i)} = \emptyset$
		\State Let $H^{d \times d}$ be a Hadamard matrix \Comment{With $d = 2^{\lceil \log_2 (n-i+1) \rceil}$}
		\State{$H^{d \times (n-i+1)} \gets n-i+1 \text{ columns of } H^{d \times d}$} \Comment{It does not matter which columns}
		\For{all $\vec r$ such that $1 \leq r_1 < \dots < r_i \leq n$} \Comment{There are $\binom{n}{i}$ such $\vec r$}
			\State $H^{d \times n}_{\vec r} \gets$ duplicate one column of $H^{d \times (n-i+1)}$ $i-1$ times to indices $r_1, \dots , r_i$
			\State $\mathcal E_n^{(i)} \gets \mathcal E_n^{(i)} \cup \Set*{ \vec m \vec m^T \given \vec m \text{ is a row of } H^{d \times n}_{\vec r} }$
		\EndFor
		\State $\mathcal E_n^{[j]} \gets \mathcal E_n^{[j]} \cup \mathcal E_n^{(i)} $
		\EndFor
	\end{algorithmic}
	\raggedright\textbf{Output:} $\mathcal E_n^{[j]}$
\end{algorithm}

\section{Bounds on the total GZZ time}
\label{sec:lin_scaling}
Our main analytic result (\cref{thm:gate-time}) is that the optimal total $\Gate{GZZ}$ time $\vec 1^T \vec \lambda^*$ is lower and upper bounded by the norms $\linfnorm{ M }$ and $\lonenorm{M}$, respectively.
Note, that for a dense matrix $M$, its norm $\lonenorm{M}$ scales quadratic with the number of qubits $n$.
We conjecture an improved upper bound on the total $\Gate{GZZ}$ time for dense $M$ that scales at most linear with the number of qubits $n$.
We support this conjecture with explicit solutions for the \ac{LP} \eqref{eq:LP1} reaching this bound for any $n$ and numerical results validating the conjecture for $n \leq 8$.

\begin{theorem}\label{thm:gate-time}
 The optimal total gate time of $\Gate{GZZ}(A)$ with $A \in \HTL (\RR^n)$ is lower and upper bounded by
 \begin{equation}
  \linfnorm{ M } \leq \vec 1^T \vec \lambda^* \leq \lonenorm{M} \, ,
 \end{equation}
 where $M \coloneqq A \oslash J$.
 Equality holds for the lower bound for all matrices $M = C \vec m \vec m^T$ for any $\vec m \in \{-1 , +1 \}^n$ and $C\geq0$.
\end{theorem}
\begin{proof}
The lower bound has been shown in \cref{lem:triv_low_bound}.
Equality in the lower bound holds for $M = C \vec m \vec m^T$ by setting $\lambda_{\vec m} = C = \linfnorm{ M }$ and $\lambda_{\vec m^\prime} = 0$ for all $\vec m^\prime \neq \vec m$.
We use the explicit construction of the standard basis elements for symmetric matrices from the proof of \cref{thrm:spanning_set} to show the upper bound.
To be precise, we have
\begin{equation}
 \frac{|M_{ij}|}{d} (H^{d \times n}_{\vec r})^T H^{d \times n}_{\vec r} - d \1_{n} = M_{ij} e_{(i, j)} \, ,
\end{equation}
where $\vec r = (i,j)$ if $M_{ij} \geq 0$ or $\vec r = (\tilde i,  \tilde j)$ if $M_{ij} < 0$ as in \cref{eq:std_basis_decomp,eq:minus_std_basis_decomp}, respectively. We define $\vec \lambda^{(i,j)}$ with the entries $\lambda^{(i,j)}_{\vec m} \coloneqq \frac{|M_{ij}|}{d} \, [\vec m \in \operatorname{rows}(H^{d}_{\vec r})]$. According to \cref{thrm:tensor_GZZ} we have $\vec 1^T \vec \lambda^{(i,j)} = |M_{ij}|$.
Adding $\vec \lambda^{(i,j)}$ for all $i<j$ yields the upper bound $\lonenorm{M}$.
\end{proof}
These bounds get tighter the sparser $M$ is.
If $M$ has only one non-zero value, then clearly $\linfnorm{ M } = \vec 1^T \vec \lambda^* = \lonenorm{ M}$.
Furthermore, these bounds also hold for the heuristic, which we presented in \cref{sec:LP_approx}.

Next, we state our conjecture that the optimal gate time of $\Gate{GZZ}(A)$ scales at most linear with the number of qubits.
\begin{conjecture}\label{conj:gate-time}
 The optimal gate time of $\Gate{GZZ}(A)$ with $A \in \HTL (\RR^n)$ is tightly upper bounded by
 \begin{equation}
  \vec 1^T \vec \lambda^* \leq \linfnorm{ M } \cdot \begin{cases}
   n \, , &\text{for odd } n\, ,\\
   n-1 \, , &\text{for even } n\, .
  \end{cases}
 \end{equation}
\end{conjecture}
Hence, it provides a tighter bound for dense $M$ than \cref{thm:gate-time}.
To support the claim of \cref{conj:gate-time} we first construct explicit dual and primal feasible solutions for the case $M = -E_n$ for the \ac{LP}'s \eqref{eq:LP1} and \eqref{eq:dualLP}, respectively.
 Then optimality is given by showing equality of the objective function values of the primal and dual problem.
 We further show that the case $M = -E_n$ leads to the same objective function value as $M = - \vec m \vec m^T$ for any $\vec m\in \{-1,1\}^n$.
 Finally, we provide numerical evidence that the conjecture holds for $n\leq 8$.

For practical purposes it is important to keep in mind that the platform given $J$ matrix might also scale with the number of qubits resulting in a qubit dependent constant $\linfnorm{ A \oslash J_n} = \linfnorm{ M_n }$.

\subsection{Explicit solutions for \texorpdfstring{$M = -\vec m \vec m^T$}{M=mm*}}
\label{sec:explicit_worst_case}
The following lemma will be used in the proof of the explicit feasible solution of the dual problem for $M$ being of the form $M = -\vec m \vec m^T$.
We can identify $\vec m\in \{-1,1\}^n$ with $\vec b \in \FF_2^n$ via $\vec m = (-1)^{\vec b}$ as explained in \cref{sec:prev_work}.
\begin{lemma}
\label{lem:gram}
It holds that
 \begin{equation}
 \begin{aligned}
 P_{\vec b} \coloneqq \sum_{i<j}^n (-1)^{b_i \oplus b_j} = \binom{n}{2}-2 |\vec b | (n- |\vec b| ) \, ,
 \end{aligned}
 \end{equation}
 for any binary vector $\vec b \in \FF_2^n$.
 We denote the Hamming weight by $|\vec b|$.
\end{lemma}
\begin{proof}
Let $\vec{ m} = (-1)^{\vec b}$.
If the Hamming weight $|\vec b |$ vanishes, then $\vec{ m} = (+1, \dots , +1)$ and $P_{\vec b} = \binom{n}{2}$ which is the maximal value.
If $|\vec b | \neq 0$, $\vec{ m}$ contains $|\vec b |$ entries $-1$, such that the upper triangular part of $\vec{m} \vec{m}^T$ contains a rectangle of $-1$'s with length $| \vec b|$ and width $n-| \vec b|$ so the total amount of $-1$'s is $| \vec b|(n-| \vec b|)$.
Therefore,
\begin{equation}
 P_{\vec b} = \binom{n}{2}-2 |\vec b| (n- |\vec b | ) \, .
\end{equation}
\end{proof}

\begin{lemma}[explicit dual feasible solution]
\label{lem:lin_scal_ana_dual}
 Let $M = -E_n$, then there is an explicit feasible solution $\vec y$ to the dual \ac{LP}~\eqref{eq:dualLP} with
  \begin{equation}
  \innerp{-E_n}{\vec y} = \begin{cases}
   n \, , &\text{for odd } n \, , \\
   n-1 \, , &\text{for even } n \, .
  \end{cases}
 \end{equation}
\end{lemma}
\begin{proof}
 We assume that $y = y_1 = y_2 = \dots$.
 Therefore, it suffices to show that
 \begin{equation}
 \label{eq:dual_constraint}
  y \sum_{i<j}^n (-1)^{b_i \oplus b_j} \leq 1 \, .
 \end{equation}
From \cref{lem:gram} we know that $\vec y = 1/\operatorname{min}(P_{\vec b}) \vec 1$ satisfies the constraint in \cref{eq:dual_constraint}.
The minimum $\operatorname{min}(P_{\vec b}) = - \lfloor n/2 \rfloor$ is reached for $|\vec b| = \lceil n/2 \rceil$ or $|\vec b| = \lfloor n/2 \rfloor$ which can be verified from the expression of $P_{\vec b}$ in \cref{lem:gram}.
Thus, we obtain $\vec y = -1/\lfloor n/2 \rfloor \vec 1$.
The objective function evaluates to
\begin{equation}
 \innerp{-E_n}{\vec y} = \vec -\vec 1^T \vec y = \frac{\binom{n}{2}}{ \lfloor n/2 \rfloor} =
   \begin{cases}
   n \, , &\text{for odd } n \, ,\\
   n-1 \, , &\text{for even } n \, .
  \end{cases}
\end{equation}
\end{proof}
For the construction of a feasible solution to the primal problem we first require the following result.
\begin{lemma}
 \label{lem:nchoosek_comb}
Let $ k < n$ be natural numbers.
Then
 \begin{equation}
 \label{eq:bi_opplus_bj}
  \sum_{\vec b \in \FF_2^n: | \vec b | = k} |b_i \oplus b_j | = 2 \binom{n-2}{ k-1 }
 \end{equation}
 for all $i,j \in [n]$ with $i\neq j$.
\end{lemma}
\begin{proof}
 Consider the case $n=4$ and $k=2$.
 We get $\binom{4}{2} = 6$ binary vectors with $| \vec b | = 2$.
 It can be easily verified that $\sum_{| \vec b | = 2}|b_i \oplus b_j | = 4 = 2 \binom{2}{ 1 }$.
 Now, we assume that for a given $n$ and $k < n$ the \cref{eq:bi_opplus_bj} holds.
 It suffices to verify \cref{eq:bi_opplus_bj} for $i \leq n$ and $j = n+1$.
 We fix $k$, define $n^\prime \coloneqq n+1$ and take a $\vec b \in \FF_2^{n^\prime}$.
 We have $\binom{n+1}{ k }$ binary vectors with $| \vec b | = k$.

 For the case $b_{n+1} = 0$ we have $\binom{n}{ k }$ such vectors and
 \begin{equation}
  \sum_{\substack{| \vec b | = k \\ b_{n+1} = 0 } }|b_i \oplus b_{n+1} | = \sum_{| \vec b | = k }|b_i| = \binom{n-1}{ k-1 } \, ,
 \end{equation}
 for all $i \leq n$.
 There are $\binom{n-1}{ k-1 }$ different ways in placing $k-1$  $1$'s.

 For the case $b_{n+1} = 1$ we have $\binom{n}{ k-1 }$ such vectors and
  \begin{equation}
  \sum_{\substack{| \vec b | = k \\ b_{n+1} = 1 } }|b_i \oplus b_{n+1} | = \sum_{| \vec b | = k }|b_i \oplus 1 | = \binom{n-1}{ k-1 } \, ,
 \end{equation}
 for all $i \leq n$.
 There are $\binom{n-1}{ k-1 }$ different ways in placing $k-1$  $0$'s.

 Combining the two cases $b_{n+1} = 0$ and $b_{n+1} = 1$ we obtain
 \begin{equation}
  \sum_{| \vec b | = k } |b_i \oplus 0 |+|b_i \oplus 1 | = 2 \binom{n-1}{ k-1 } = 2 \binom{n^\prime -2}{ k-1 } \, ,
 \end{equation}
 for all $i \leq n$.
\end{proof}

We motivate the next lemma with the result of the explicit dual feasible solution for $M = -E_n$ from \cref{lem:lin_scal_ana_dual} and the complementary slackness condition.
The complementary slackness condition for a \acl{LP} states that if the $i$-th inequality of the dual problem is a strict inequality for a feasible solution $\vec y$, then the $i$-th component of a feasible solution of the primal problem $\vec \lambda$ is zero:
\begin{equation}
\label{eq:compl_slack}
 (V^T \vec y)_i < 1 \Rightarrow \lambda_i = 0 \, .
\end{equation}
We use this in the following lemma to construct a feasible solution for the primal \ac{LP}~\eqref{eq:LP1}.
\begin{lemma}[explicit primal feasible solution]
\label{lem:lin_scal_ana_prim}
 Let $M = -E_n$, then there is an explicit feasible solution $\vec \lambda$ to the primal \ac{LP}~\eqref{eq:LP1} with
  \begin{equation}
  \vec 1^T \vec \lambda = \begin{cases}
   n \, , &\text{for odd } n \, ,\\
   n-1 \, , &\text{for even } n \, .
  \end{cases}
 \end{equation}
\end{lemma}
\begin{proof}
For this proof we define
\begin{equation}
 k \coloneqq \begin{cases}
   n \, , &\text{for odd } n \, , \\
   n-1 \, , &\text{for even } n \, .
  \end{cases}
\end{equation}
We only consider binary vectors of the set $S \coloneqq \Set*{\vec b \in \FF_2^n \given |\vec b| = \lceil n/2 \rceil, \lfloor n/2 \rfloor \text{ and } b_n=0}$.
This set is motivated by the complementary slackness condition and \cref{lem:lin_scal_ana_dual}.
It can be calculated that $|S| = \binom{k}{ \lfloor k/2 \rfloor }$ using the recurrence relation
\begin{equation}
 \binom{n}{k} = \binom{n-1}{k} + \binom{n-1}{k-1}
\end{equation}
for the binomial coefficients.
We show that
\begin{equation}
\label{eq:sum_prim_sol}
 \sum_{\vec b \in S} (-1)^{b_i \oplus b_j} = - D_n \, ,
\end{equation}
for a constant $D_n>0$, which we calculate later.
If not explicitly stated, all equations in this proof containing $i$, $j$ hold for all $i, j \in [n]$, $i \neq j$.
We denote $\vec \lambda_S$ by all $\lambda_{\vec m} = \lambda_{\vec b}$ corresponding to the encoding $\vec m = (-1)^{\vec b}$ with $\vec b \in S$.
If \cref{eq:sum_prim_sol} holds, we can choose a $\vec \lambda_S = 1 / D_n \vec 1$ resulting in
\begin{equation}
\label{eq:sum_prim_sol_minusOne}
 \sum_{\vec b \in S} \vec \lambda_{\vec b} (-1)^{b_i \oplus b_j} = -1 \, ,
\end{equation}
which implies feasibility for $M = - E_n$.
It is left to show \cref{eq:sum_prim_sol} and determine $D_n$.

First, we consider odd $n$.
By definition of $S$ we have that $b_n = 0$ for all $\vec b$.
Therefore, we obtain $\binom{n-1}{ \lfloor n/2 \rfloor } + \binom{n-1}{ \lceil n/2 \rceil }$ binary vectors with $|\vec b| = \lceil n/2 \rceil$ or $|\vec b| = \lfloor n/2 \rfloor$ respectively. Counting the occurrences of ``$-1$'' in the sum of \cref{eq:sum_prim_sol} is equivalent to counting the occurrences of ``$1$'' in the sum
\begin{equation}
\begin{aligned}
 \sum_{\vec b \in S} |b_i \oplus b_j| &= 2 \left( \binom{n-3}{ \lfloor \frac n2 \rfloor -1 } + \binom{n-3}{ \lceil \frac n2 \rceil -1 } \right) \\
 &= 2 \binom{n-2}{ \lfloor \frac n2 \rfloor } \, ,
\end{aligned}
\end{equation}
where we used \cref{lem:nchoosek_comb}, the recurrence relation for the binomial coefficients and $\lceil n/2 \rceil -1 = \lfloor n/2 \rfloor$.
Counting the occurrences of ``$+1$'' in the sum of \cref{eq:sum_prim_sol} yields
\begin{equation}
 \binom{n-1}{ \lfloor \frac n2 \rfloor } + \binom{n-1}{ \lceil \frac n2 \rceil } - 2 \binom{n-2}{ \lfloor \frac n2 \rfloor } = \binom{n}{ \lfloor \frac n2 \rfloor } - 2 \binom{n-2}{ \lfloor \frac n2 \rfloor } \, .
\end{equation}
where we used $\binom{n}{ \lfloor n/2 \rfloor +1 } = \binom{n}{ \lfloor n/2 \rfloor }$ for odd $n$.
We now evaluate \cref{eq:sum_prim_sol} for odd $n$
\begin{equation}
\label{eq:sum_prim_sol_odd}
\begin{aligned}
 -D_n = \sum_{\vec b \in S} (-1)^{b_i \oplus b_j} &= \binom{n}{ \lfloor \frac n2 \rfloor } - 4 \binom{n-2}{ \lfloor \frac n2 \rfloor } \\
 &= \binom{n}{ \lfloor \frac n2 \rfloor } - \frac{n+1}{n} \binom{n}{ \lfloor \frac n2 \rfloor } \\
 &= \frac{-1}{n} \binom{n}{ \lfloor \frac n2 \rfloor } \, .
\end{aligned}
\end{equation}
The case for even $n$ follows the same steps as for odd $n$, resulting in
\begin{equation}
\label{eq:prim_sol_even_minus}
\begin{aligned}
 -D_n = \sum_{| \vec b | =  \frac n2 } (-1)^{b_i \oplus b_j} &=\binom{n-1}{ \frac n2 } - \frac{n}{n-1} \binom{n-1}{ \frac n2 } \\
 &= \frac{-1}{n-1} \binom{n-1}{ \frac n2 } \, .
\end{aligned}
\end{equation}
\Cref{eq:sum_prim_sol_odd,eq:prim_sol_even_minus} show that $D_n = 1/k \binom{k}{ \lfloor k/2 \rfloor }$ and that \cref{eq:sum_prim_sol} holds.
Since $|S| = \binom{k}{ \lfloor k/2 \rfloor }$ the objective function value is
\begin{equation}
 \vec 1^T \vec \lambda_S = \binom{k}{ \lfloor k/2 \rfloor } \frac{1}{\frac 1k \binom{k}{ \lfloor k/2 \rfloor }} = k \, .
\end{equation}
\end{proof}
From the equality of the objective functions for the primal and dual problem from \cref{lem:lin_scal_ana_prim} and \cref{lem:lin_scal_ana_dual} respectively we know that the proposed dual/primal feasible solutions for $M = -E_n$ are in fact optimal solutions.
Now, we show that the $\Gate{GZZ}$ gate time with $M = -\vec m \vec m^T$ for any $\vec m \in \{-1 , +1 \}^n$ is the same as for the case $M = -E_n$.
\begin{theorem}
\label{prop:worst_case_sol}
 If $A \oslash J \eqqcolon M = -C (\vec m \vec m^T)$ for any $\vec m \in \{-1 , +1 \}^n$ and $C\geq0$, then the optimal gate time of $\Gate{GZZ}(A)$ is
 \begin{equation}
  \vec 1^T \vec \lambda^* = \linfnorm{ M } \begin{cases}
   n \, , &\text{for odd } n \, ,\\
   n-1 \, , &\text{for even } n \, .
  \end{cases}
 \end{equation}
\end{theorem}
\begin{proof}
 The statement has been shown for the case $M = -E_n$ by constructing an explicit solution.
 It is left to show that the cases $M = -C (\vec m \vec m^T)$ for a constant $C\geq0$ yield the same objective function value.
 Since the objective function and the constraints are linear we can w.l.o.g.\ assume $C=1$.
 It is clear, that $\operatorname{sign}(M_{ij}) = \operatorname{sign}(-m_i m_j) = -(-1)^{c_i \oplus c_j}$ for all $i<j$, with $\vec m = (-1)^{\vec c}$.
Let $\vec y = -y \operatorname{sign}(\vec v(M))$ for a $y \in \R$, then each constraint of $V^T \vec y \leq \vec 1$ of the dual \ac{LP}~\eqref{eq:dualLP} reads as
\begin{equation}
 y \sum_{i<j} (-1)^{b_i \oplus b_j \oplus c_i \oplus c_j} = y \sum_{i<j} (-1)^{\tilde b_i \oplus \tilde b_j } \leq 1 \, ,
\end{equation}
with $\vec{\tilde b} \coloneqq \vec b \oplus \vec c$ element wise.
Consider the ordered set of all $\vec{b} \in \FF_2^n$ with $b_n = 0$, then $\vec{\tilde b}$ is just a permutation of that set.
Due to the permutation symmetry of the qubits the optimal value of the \ac{LP}~\eqref{eq:LP1} for any $M = - (\vec m \vec m^T)$ is the same as for the case $M = - E_n$.
Setting $\vec y = -y \operatorname{sign}(\vec v(M))$ with $y=1/\lfloor n/2 \rfloor$ as in the proof of \cref{lem:lin_scal_ana_dual} yields an optimal solution to the dual \ac{LP}~\eqref{eq:dualLP} for the case $M = - (\vec m \vec m^T)$.
\end{proof}
Note that, trivially, the lower bound $\vec 1^T \vec \lambda^* = \linfnorm{ M }$ is reached if $M = C (\vec m \vec m^T)$ for any $\vec m \in \{-1 , +1 \}^n$ and $C\geq0$.

One possibility to prove \cref{conj:gate-time} is to show, that the matrix $M = -E_n$ maximizes the value of the \ac{LP}~\eqref{eq:LP1} among all matrices $M \in \HTL([-1, +1]^n)$.
To this end, consider the \ac{LP}
\begin{maxi}
{}{ \lonenorm{ \vec y }}{}{}
\label{eq:upper_bound_LP}
\addConstraint{V^T \vec y } {\leq \vec 1}{}{}
\addConstraint{\vec y}{\in \R^{\binom{n}{2}}}{}{} \, ,
\end{maxi}
which is independent of $M \in \HTL([-1, +1]^n)$.
It holds that
\begin{equation}
 \max_{V^T \vec y \leq \vec 1} \left( \max_{\linfnorm{M} \leq 1} \innerp{M}{\vec y} \right) = \max_{V^T \vec y \leq \vec 1} \lonenorm{ \vec y }
\end{equation}
according to $\lonenorm{ \vec x } = \max_{\linfnorm{\vec p} \leq 1} \vec p^T \vec x $.
Therefore, the optimal objective value of \ac{LP}~\eqref{eq:upper_bound_LP} is an upper bound on all optimal objective values of \ac{LP}~\eqref{eq:LP1}.
Clearly, the constructed solution in \cref{lem:lin_scal_ana_dual} is feasible for \ac{LP}~\eqref{eq:upper_bound_LP}.
Unfortunately, proving that this constructed solution is optimal is quite challenging, as we discuss in \cref{app:proof_problems}.

\begin{figure}[H]
	\centering
    \includegraphics[width=0.52\linewidth]{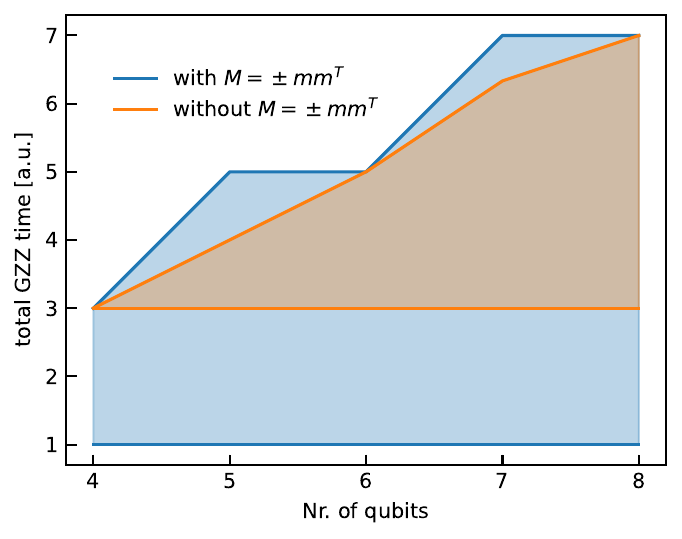}
	\caption{
			The optimal objective function value of the dual \ac{LP}~\eqref{eq:dualLP} over the number of qubits.
			\textbf{Blue:} The range of optimal values for all binary $M \in \HTL(\{-1, +1\}^n)$.
			\textbf{Orange:} The range of optimal values for all binary $M \in \HTL(\{-1, +1\}^n)$ without $M = \pm (\vec m \vec m^T)$ for any $\vec m \in \{-1 , +1 \}^n$.
	}
	\label{fig:worst_case_GZZ}
\end{figure}

\section{Numerical results}
\label{sec:numerics}
\subsection{Numeric validation of Conjecture \ref{conj:gate-time} for small \texorpdfstring{$n$}{n}}
For the numeric validation of the conjecture for small $n$ we solve the dual \ac{LP}~\eqref{eq:dualLP} for all binary $M \in \HTL(\{-1, +1\}^n)$ of which there are $2^{\binom{n}{2}}$.
For $n=3$ there are only binary matrices of the form $M = \pm (\vec m \vec m^T)$ and by \cref{prop:worst_case_sol} the conjecture holds.
\Cref{fig:worst_case_GZZ} shows that \cref{conj:gate-time} holds for $n\leq 8$.
For odd $n \leq 8$ the cases $M = - (\vec m \vec m^T)$ are in fact the only cases reaching the upper bound.
This can be seen in \cref{fig:worst_case_GZZ} by the blue area exceeding the orange area, which only consists of the optimal values for all binary matrices without $M = \pm (\vec m \vec m^T)$ for any $\vec m \in \{-1 , +1 \}^n$.

\subsection{Numerical benchmark for the heuristic}

\begin{figure}[t]
	\centering
    \includegraphics[width=0.496\linewidth]{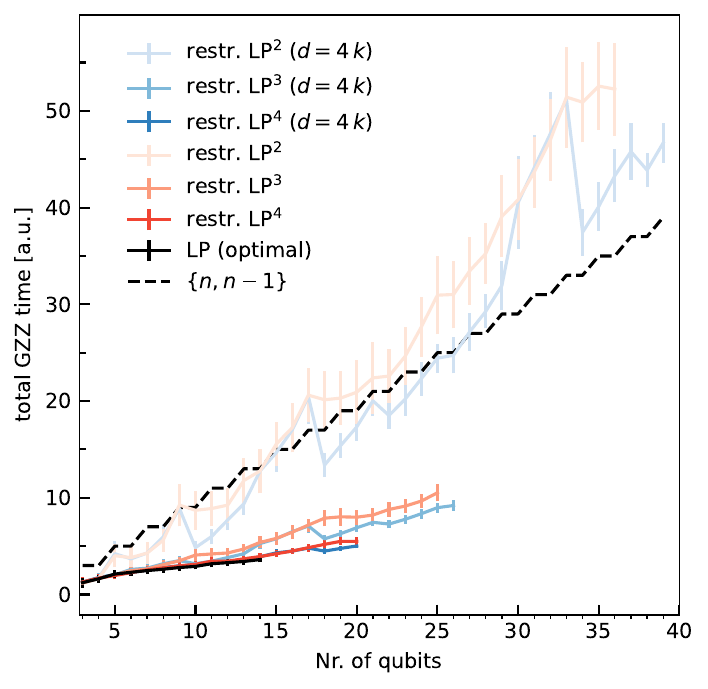}
    \includegraphics[width=0.496\linewidth]{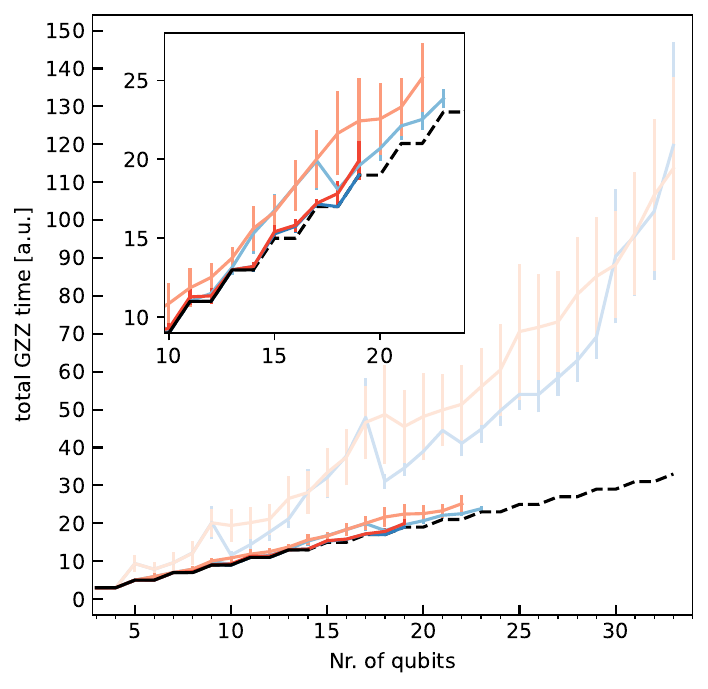}
	\caption{
			Comparing the performance of the original (optimal) \ac{LP}~\eqref{eq:LP1} and the restricted \ac{LP}$^j$~\eqref{eq:LP_poly_hierarchy} for $j = 2,3,4$.
			For each line we let the \ac{LP}'s run for a fixed time.
			The black dashed line is the upper bound for the original \ac{LP}~\eqref{eq:LP1} from \cref{conj:gate-time}.
			The reddish lines show the total $\Gate{GZZ}$ times for the restricted \ac{LP}$^j$~\eqref{eq:LP_poly_hierarchy} for $j = 2,3,4$.
			The blueish lines show the total $\Gate{GZZ}$ times for the restricted \ac{LP}$^j$~\eqref{eq:LP_poly_hierarchy} using Hadamard matrices of dimension $d=4\,k$ to construct the restricted linear operator $V^{[j]}$.
			\textbf{Left:} Average case scaling of the total $\Gate{GZZ}$ times.
			The data points and error bars show the mean and the standard deviation over $100$ uniformly sampled matrices $A_{ij} = A_{ji} \in [-1, 1]$ for $i<j$.
			\textbf{Right:} Conjectured worst-case scaling of the total $\Gate{GZZ}$ times.
			The data points and error bars show the mean and the standard deviation over $100$ binary matrices $A = - \vec{m} \vec{m}^T$ with uniformly sampled encodings $\vec{m} \in \{-1,1\}^n$.
	}
	\label{fig:eff_GZZ}
\end{figure}

We compare the performance of the restricted \ac{LP}$^j$~\eqref{eq:LP_poly_hierarchy} to the original \ac{LP}~\eqref{eq:LP1}.
To this end we provide numerical results on the average-case and the conjectured worst-case scaling of the total $\Gate{GZZ}$ time.
The left plot in \cref{fig:eff_GZZ} shows the average-case scaling of the total $\Gate{GZZ}$ time.
$\Gate{GZZ}(A)$ gates with uniformly sampled matrix elements $A_{ij} = A_{ji} \in [-1, 1]$ for all $i<j$ and $i,j \in [n]$ are synthesized.
For the worst-case scaling of the right plot in \cref{fig:eff_GZZ}, $\Gate{GZZ}(A)$ gates with binary matrices $A = - \vec{m} \vec{m}^T$ with uniformly sampled encodings $\vec{m} \in \{-1,1\}^n$ are synthesized.
We note that these are the matrices with the conjectured worst-case scaling for the \ac{LP} and that the \ac{LP}$^j$ might have different worst-case matrices $A$.
For convenience, we have set $M = A$, i.e.\ setting $J = E_n$ which omits the hardware specific time units given by the quantum platform.
For realistic $J$'s and $\Gate{GZZ}(A)$ gate times we refer to \cite{basler_synthesis_2022}.
The Python package CVXPY \cite{diamond2016cvxpy,agrawal2018rewriting} with the GNU linear program kit simplex solver \cite{GLPK} is used to solve the \ac{LP}~\eqref{eq:LP1} and the restricted \ac{LP}$^j$~\eqref{eq:LP_poly_hierarchy}.

\Cref{fig:eff_GZZ} shows the total $\Gate{GZZ}$ time over the number of qubits $n$ for the restricted \ac{LP}$^j$~\eqref{eq:LP_poly_hierarchy} and \ac{LP}~\eqref{eq:LP1}.
We assigned a fixed runtime ($20$ minutes) for each \ac{LP} to synthesize $\Gate{GZZ}(A)$ gates for all sample matrices $A$ and as many $n$ as possible.
Clearly, the runtime of the heuristic algorithm is much shorter than the runtime of the optimal \ac{LP}~\eqref{eq:LP1}.
Increasing the hierarchy of the heuristic $j>2$ reduces the total $\Gate{GZZ}$ time significantly while still maintaining a short runtime.
The total $\Gate{GZZ}$ time obtained from the heuristic also seems to scale linear with the number of qubits although with a different scaling constant.
As mentioned in \cref{sec:ana_GZZ,sec:LP_approx} the size of the restricted \ac{LP}$^j$~\eqref{eq:LP_poly_hierarchy}, and therefore the runtime, can be further reduced.
This reduction is achieved by using Hadamard matrices of dimension $d = 4\,k$ with $k \in \NN$ instead of Sylvester-Hadamard matrices of dimension $d = 2^k$ in the construction of the restricted linear operator $V^{[j]}$.
This is based on the Hadamard conjecture, which is known to be true for $d \leq 668$ \cite{kharaghani_2005,dokovic_2014}.
Surprisingly, using these Hadamard matrices with $d=4\,k$ not only yields shorter runtime of the heuristic algorithm but also a significant reduction of the total $\Gate{GZZ}$ time compared to the original restricted \ac{LP}$^j$~\eqref{eq:LP_poly_hierarchy}.

Our numerical results show that the heuristic algorithm approximates well the optimal total $\Gate{GZZ}$ time, while maintaining a short runtime.
This holds true for both, the average-case and the conjectured worst-case scaling.
Therefore, we hope that this heuristic will prove to be an important tool to implement fast $\Gate{GZZ}$ gates in practice.

\section{Conclusion}
We investigated the time-optimal multi-qubit gate synthesis introduced in Ref.~\cite{basler_synthesis_2022}.
We show that synthesizing time-optimal multi-qubit gates in our setting is \NP-hard.
However, we also provide explicit solutions for certain cases with constant gate time and a polynomial-time heuristics to synthesize fast multi-qubit gates.
Our numerical simulations suggest that these heuristics provide good approximations to the optimal $\Gate{GZZ}$ gate time.
Furthermore, tight bounds on the scaling of the optimal multi-qubit gate times were shown.
More precisely, we showed that the optimal multi-qubit gate time scales at most as $\lonenorm{ A \oslash J }$, the $\ell_1$-norm of the element-wise division of the total and physical coupling matrices $A$ and $J$, respectively.
We also conjectured that the optimal $\Gate{GZZ}$ gate time scales at most linear with the number of qubits.
Our results are practical to estimate the execution time of a given circuit, where the entangling gates are implemented as $\Gate{GZZ}$ gates.
The execution time is a crucial parameter, in particular, in the \ac{NISQ} era since it is limiting the length of a gate sequence due to finite coherence time.

It is our hope to proof the conjectured linear scaling of the optimal $\Gate{GZZ}$ gate time in the near future.
Moreover, we would like to test and verify our proposed time-optimal multi-qubit gate synthesis methods in an experiment.
Depending on the quantum platform we would like to develop adapted error mitigation schemes for the $\Gate{GZZ}$ gates and investigate their robustness against errors.

\section*{Acknowledgements}
We are grateful to Lennart Bittel and Arne Heimendahl for fruitful discussions on complexity theory and convex optimization, respectively.
We also want to thank Frank Vallentin for valuable comments on our conjecture and proof ideas.

This work has been funded by the German Federal Ministry of Education and Research (BMBF) within the funding program ``quantum technologies -- from basic research to market'' via the joint project MIQRO (grant number 13N15522) and 
the Fujitsu Services GmbH as part of an endowed professorship ``Quantum Inspired and Quantum Optimization''. 

\newpage
\section*{Appendices}
\appendix

\section{Challenges in proving Conjecture \ref{conj:gate-time}}
\label{app:proof_problems}
In this section, we want to discuss some obstacles encountered trying to proof \cref{conj:gate-time}.
\cref{conj:gate-time} holds, if we show that our constructed solutions from \cref{sec:explicit_worst_case} are optimal solutions for \ac{LP}~\eqref{eq:upper_bound_LP}.
Meaning that there is no other feasible solution resulting in a larger objective function value compared to our constructed solution.

We tried an inductive proof which turned out to be intricate due to the additional degrees of freedom in each induction step.
Furthermore, we utilized the connection to graph theory from \cref{lem:cut_ellip} to transform the \ac{LP}~\eqref{eq:upper_bound_LP} to a \ac{LP} over the cut polytope.
There, the challenge is the affine mapping between the elliptope and the cut polytope in \cref{def:outprod} and \cref{def:cut_polytope}, which alters the optimization problem crucially.
In the following, we discuss another approach based on showing the sufficiency of the \ac{KKT} conditions in more detail.

\subsection{Concave program}
A convex linear program in standard form minimizes a convex objective function over a convex set.
But \ac{LP}~\eqref{eq:upper_bound_LP} \emph{maximizes} a convex objective function over a convex set.
Such optimizations are called \emph{concave programs}.
It is known that the maximum is attained at the extreme points of the polytope $V^T \vec y \leq \vec 1$ and therefore might have many local optima \cite{Phillips1988}.
There are several equivalent sufficient conditions for optimality \cite{dur_necessary_1998}.
Here, we investigate one in detail.
First, we define the \emph{conjugate function} for a function $f: \R^n \rightarrow \R$
\begin{equation}
 f^* (\vec x) \coloneqq \sup \Set*{ \vec x^T \vec y - f(\vec y) \given \vec y \in \mathcal D (f)} \, ,
\end{equation}
with $\mathcal D (f)$ the domain of $f$.
Furthermore, we need the \emph{support function} $h_A: \R^n \rightarrow \R$ for a closed convex set $A$
\begin{equation}
 h_A(\vec x ) \coloneqq \sup \Set*{ \vec x^T \vec y \given  \vec y \in A} \, .
\end{equation}
Then the sufficient optimality condition in our case is
\begin{equation}
 \lonenorm{ \vec y^* } = \sup \Set*{ h_A(\vec x ) - (\lonenorm{ \vec x })^* \given  \vec y \in \R^{\binom{n}{2}} } \, ,
\end{equation}
with $A = \Set*{ \vec x \in \R^{\binom{n}{2}} \given  V^T \vec x \leq \vec 1}$.
The conjugate function of $\lonenorm{ \vec x }$ is
\begin{equation}
 (\lonenorm{ \vec x })^* = \begin{cases}
   0 \, , &\text{if } \linfnorm{ \vec x } \leq 1 \, ,\\
   \infty \, , &\text{otherwise} \, .
  \end{cases}
\end{equation}
Then we have
\begin{equation}
\begin{aligned}
 \lonenorm{ \vec y^* } &= \sup \Set*{ h_A (\vec y)  \given  \linfnorm{ \vec y } \leq 1} \\
 &= \sup \Set*{ \vec x^T \vec y  \given  \linfnorm{ \vec y } \leq 1, V^T \vec x \leq \vec 1 } \\
 &= \sup \Set*{ \lonenorm{ \vec x } \given V^T \vec x \leq \vec 1 } \, ,
\end{aligned}
\end{equation}
which is the same formulation as the original \ac{LP}~\eqref{eq:upper_bound_LP}.

\subsection{Dualization}
If we can formulate the dual \ac{LP} to the primal \ac{LP}~\eqref{eq:upper_bound_LP} and find a feasible solution, then the dual objective function value upper bounds the primal objective function value by weak duality \cite{bazaraa2013nonlinear}.
The standard form of an optimization problem with only linear inequality constraints is
\begin{mini}
{}{ f(\vec y)}{}{}
\addConstraint{A \vec y } {\leq\vec b}{}{}.
\end{mini}
Then, the Lagrange dual function is given by
\begin{equation}
 g(\vec \lambda) = - \vec b^T \vec \lambda - f^* (-A^T \vec \lambda) \, ,
\end{equation}
where $f^*$ denotes the conjugate function \cite{Boyd2009convexoptimization}.
For \ac{LP}~\eqref{eq:upper_bound_LP} we have $f(\vec y ) = - \lonenorm{ \vec y }$ (minus sign due to the minimization in the standard optimization form).
\begin{equation}
\begin{aligned}
 (- \lonenorm{ \vec x })^* &= \sup \Set*{ \vec x^T \vec y + \lonenorm{ \vec x } \given \vec y \in \R^{\binom{n}{2}}}  \\
 &= \sup \Set*{ \vec x^T \vec y + \vec x^T \vec z \given \vec y \in \R^{\binom{n}{2}}, \linfnorm{\vec z} \leq 1} \\
 &= \sup \Set*{ \vec x^T \vec y \given \vec y \in \R^{\binom{n}{2}}} \, ,
\end{aligned}
\end{equation}
which clearly is unbounded if $\vec x \neq \vec 0$.
Therefore, we cannot formulate the dual to \ac{LP}~\eqref{eq:upper_bound_LP}.

\subsection{Invexity}
The \ac{KKT} conditions are optimality conditions for non-linear optimization problems.
Invexity is a generalization of convexity in the sense that the \ac{KKT} conditions are necessary and sufficient for optimality \cite{hanson_invexity_1999}.
By invexity we mean that the objective and constraint functions of the optimization problem are \emph{Type 1 invex} functions.
\begin{definition}
Consider the standard form of an optimization problem
\begin{mini}
{}{ f(\vec y)}{}{}
\label{eq:invexity}
\addConstraint{g( \vec y ) } {\leq\vec 0}{}{}
\addConstraint{ \vec y } {\in \mathcal S}{}{},
\end{mini}
where $\mathcal S \subseteq \R^m$ is defined by $g( \vec y ) \leq\vec 0$.
Then $f$ and $g$ are called \emph{Type 1 invex} functions at point $\vec y^* \in \mathcal S$ w.r.t.\ a common function $\eta (\vec y, \vec y^*) \in \R^m$, if for all $\vec y \in \mathcal S$,
\begin{equation}
 \begin{aligned}
  f(\vec y) -f(\vec y^*) &\geq \eta (\vec y, \vec y^*)^T \nabla f(\vec y^*), \\
  -g(\vec y^*) &\geq \eta (\vec y, \vec y^*)^T \nabla g(\vec y^*)
 \end{aligned}
\end{equation}
hold.
It suffices to consider only the active constraints, i.e.\ the constraints, where equality holds $g( \vec y^*)=0$. \cite{hanson_invexity_1999}
\end{definition}
Let $K$ be the scaling factor of the conjectured upper bound, i.e.\
\begin{equation}
 K \coloneqq \begin{cases}
   n \, , &\text{for odd } n \, ,\\
   n-1 \, , &\text{for even } n \, .
  \end{cases}
\end{equation}
In the case of \ac{LP}~\eqref{eq:upper_bound_LP} we have $\mathcal S = \Set*{ \vec y \in \R^{\binom{n}{2}} \given V^T \vec y \leq \vec 1}$ and want to show that $\vec y^* = - \vec 1 / \lfloor n/2 \rfloor$ is a global optimum.
Furthermore, we have
\begin{equation}
 \begin{aligned}
  f(\vec y) &= - \lonenorm{ \vec y }\, , \hspace{0.95em} f(\vec y^*) = -K \, , \, \,  \nabla f(\vec y^*) = \vec 1 \, , \\
  g(\vec y) &= V^T \vec y - \vec 1 \, , \, \, g(\vec y^*) = 0 \, , \hspace{1.65em} \nabla g(\vec y^*) = V|_{a} \, ,
 \end{aligned}
\end{equation}
where $V|_{a}$ are the columns of $V$ such that $(V|_{a})^T \vec y^* = \vec 1$, i.e.\ the active constraints.
To show invexity we have to find a common $\eta (\vec y, \vec y^*) \in \R^m$ such that
\begin{equation}
 \begin{aligned}
 \lonenorm{ \vec y } &\leq K - \eta^T  (\vec y, \vec y^*) \vec 1  \,  \text{ and} \\
 (V^T|_{a}) \eta  (\vec y, \vec y^*)  &\leq \vec 0 \, ,
 \end{aligned}
\end{equation}
for all $\vec y \in \mathcal S$.
It is quite challenging to find a $\eta (\vec y, \vec y^*) \in \R^m$ satisfying both inequalities.
In particular, ansätze motivated from the geometry in small dimensions eventually fail for larger $n$.

\section{Acronyms}

\begin{acronym}[MAGIC]\itemsep.5\baselineskip
\acro{AGF}{average gate fidelity}
\acro{AQFT}{approximate quantum fourier transform}

\acro{BOG}{binned outcome generation}

\acro{CP}{completely positive}
\acro{CPT}{completely positive and trace preserving}
\acro{CS}{compressed sensing} 

\acro{DAQC}{digital-analog quantum computing}
\acro{DFE}{direct fidelity estimation} 
\acro{DM}{dark matter}
\acro{DD}{dynamical decoupling}

\acro{EF}{extended formulation}

\acro{GST}{gate set tomography}
\acro{GUE}{Gaussian unitary ensemble}

\acro{HOG}{heavy outcome generation}

\acro{IQP}{Instantaneous Quantum Polynomial
time}

\acro{KKT}{Karush–Kuhn–Tucker}

\acro{LDP}{Lamb-Dicke parameter}
\acro{LP}{linear program}

\acro{MAGIC}{magnetic gradient induced coupling}
\acro{MBL}{many-body localization}
\acro{MIP}{mixed integer program}
\acro{ML}{machine learning}
\acro{MLE}{maximum likelihood estimation}
\acro{MPO}{matrix product operator}
\acro{MPS}{matrix product state}
\acro{MS}{M{\o}lmer-S{\o}rensen}
\acro{MUBs}{mutually unbiased bases} 
\acro{mw}{micro wave}

\acro{NISQ}{noisy and intermediate scale quantum}

\acro{POVM}{positive operator valued measure}
\acro{PVM}{projector-valued measure}

\acro{QAOA}{quantum approximate optimization algorithm}
\acro{QFT}{quantum Fourier transform}
\acro{QML}{quantum machine learning}
\acro{QMT}{measurement tomography}
\acro{QPT}{quantum process tomography}
\acro{QPU}{quantum processing unit}

\acro{RDM}{reduced density matrix}

\acro{SFE}{shadow fidelity estimation}
\acro{SIC}{symmetric, informationally complete}
\acro{SPAM}{state preparation and measurement}

\acro{RB}{randomized benchmarking}
\acro{rf}{radio frequency}
\acro{RIC}{restricted isometry constant}
\acro{RIP}{restricted isometry property}

\acro{TT}{tensor train}
\acro{TV}{total variation}

\acro{VQA}{variational quantum algorithm}

\acro{VQE}{variational quantum eigensolver}

\acro{XEB}{cross-entropy benchmarking}

\end{acronym}

\bibliographystyle{./myapsrev4-2}
\bibliography{new,mk}
\end{document}